\documentclass{amsart}
\pdfoutput=1

\usepackage{amssymb}
\usepackage{amsthm}
\usepackage{amsfonts}
\usepackage{amsmath}
\usepackage{bbm}

\usepackage{pmboxdraw}
\usepackage{verbatim} 
\usepackage{graphicx}
\usepackage{color} 
\usepackage[colorlinks=true, citecolor=blue, filecolor=black, linkcolor=black, urlcolor=black]{hyperref}
\usepackage{cite}
\usepackage[normalem]{ulem}
\usepackage{subcaption}
\usepackage{tabulary}
\newcolumntype{K}[1]{>{\centering\arraybackslash}p{#1}}

\usepackage{mathtools}
\usepackage{kantlipsum}
\allowdisplaybreaks
\usepackage{enumitem} 
\usepackage{tikz}
\usepackage{pgfplots}
\pgfplotsset{compat=1.18}
\usepackage{todonotes}
\usepackage[table]{xcolor}
\usepackage{graphicx}

\usetikzlibrary{calc}
\usetikzlibrary{decorations.pathreplacing}

\newcommand{\RN}[1]{%
	\textup{\uppercase\expandafter{\romannumeral#1}}%
}

\def\wh{\widehat}

\def\R{\mathbb{R}}

\newcommand{\re}{\operatorname{Re}}

\newcommand{\ud}{\, \mathrm{d}}

\newcommand{\sgn}{\operatorname{sgn}}

\usepackage{tcolorbox}
\tcbuselibrary{most}
\newtcolorbox{defn}[1]{breakable,
colbacktitle=gray!50!white, fonttitle=\bfseries, coltitle=black, title=Definition: {#1}}

\theoremstyle{plain}
\newtheorem{thm}{Theorem}[section]

\newtheorem{lem}[thm]{Lemma}

\newtheorem{prop}[thm]{Proposition}

\theoremstyle{remark}
\newtheorem{rem}{Remark}

\numberwithin{equation}{section}

\begin{document}

\title[Half-space constrained Riesz gas]{
Confinement Transitions in Half-space constrained Riesz Gases}

\author{Sung-Soo Byun}
\address{Department of Mathematical Sciences and Research Institute of Mathematics, Seoul National University, Seoul 08826, Republic of Korea}
\email{sungsoobyun@snu.ac.kr}

\author{Yong-Woo Lee}
\address{Department of Mathematical Sciences, Seoul National University, Seoul 08826, Republic of Korea}
\email{hellowoo@snu.ac.kr}

\author{Eui Yoo}
\address{Department of Mathematical Sciences, Seoul National University, Seoul 08826, Republic of Korea}
\email{yysh0227@snu.ac.kr}

\begin{abstract} 
We study a Riesz gas with interaction parameter $s \in (d-3,d)$ in an arbitrary spatial dimension $d$, confined to a half-space by a hard wall. We prove that the equilibrium measure exhibits a dichotomy according to the interaction range. For $s \in (d-2,d)$, corresponding to the weakly long-ranged regime, the equilibrium measure always retains a non-trivial bulk component and thus is never completely confined to the wall. In contrast, for $s \in (d-3,d-2]$, corresponding to the strongly long-ranged regime, a confinement transition occurs: there exists a critical wall position beyond which the equilibrium measure is supported entirely on the wall. Our theorem generalises recent results established for the Coulomb gas, corresponding to the case $s=d-2$.
\end{abstract}

\maketitle


\section{Introduction}

Interacting particle systems arise naturally across a broad range of mathematics and physics \cite{CDFR14,ST97,Fo10}, including random matrix theory, approximation theory, statistical mechanics, and condensed matter physics. Despite their diverse origins, these systems share a unifying feature: macroscopic behaviour emerges from pairwise interactions among a large number of particles. A central objective is therefore to understand how the qualitative properties of the resulting collective states are determined by the underlying interaction law.

Among the many interaction laws that have been studied, long-range interactions occupy a particularly prominent position owing to their intrinsically nonlocal nature. Unlike short-range interactions, particles separated by macroscopic distances continue to exert a significant influence on one another, giving rise to genuinely collective phenomena that are absent in systems with only local interactions.  
Beyond the classical distinction between short-range and long-range
systems, it is also natural to ask whether long-range interactions
themselves exhibit a richer qualitative classification. 
In particular, one may ask whether different forms of nonlocal interactions produce fundamentally distinct collective behaviours, and which mathematical properties of the interaction kernel govern these differences.

These questions are naturally addressed within the framework of Riesz
gases \cite{Lew22,Se24}. For positive integers \(N\) and \(d\), let $X_N=(x_1,\ldots,x_N)\in(\mathbb{R}^d)^N$ be a point configuration.  
The $N$–particle Hamiltonian associated with Riesz gases is given by 
\begin{equation}  \label{def of Hamiltonian}
\mathcal{H}_N(X_N) := \sum_{1 \le i \ne j \le N} g_s(x_i - x_j) 
+ N \sum_{i=1}^N V(x_i),
\end{equation}
where $V:\mathbb{R}^d\to(-\infty,+\infty)$ is a confining potential and  
\begin{equation}
\label{def of interaction g}
g_s(x)  := \begin{cases}
\frac{ 1 }{s} |x|^{-s}, & s \ne 0, 
\smallskip 
\\
-\log|x|, & s = 0 , 
\end{cases} \qquad x \in \R^d. 
\end{equation} 
The interaction is conventionally classified as long-ranged when its
tail is non-integrable at infinity, and short-ranged otherwise. For the
Riesz kernel $g_s$, this distinction is determined precisely by the relation between the interaction exponent \(s\) and the dimension \(d\):
the regime \(s<d\) is long-ranged, and \(s>d\) is short-ranged; see Figure~\ref{fig:interaction-regimes}.  
We note that several normalisation conventions for the interaction kernel \eqref{def of interaction g} are used in the literature. For instance, in \cite{Lew22}, the kernel is defined as $\sgn(s)|x|^{-s}$. Our convention is chosen so as to make the continuity at \(s=0\) transparent, in view of the replica trick type limit
\begin{equation}
-\log |x|= \lim_{s\to 0} s^{-1}( |x|^{-s}-1 ) .
\end{equation}

By definition, the Riesz gas at inverse temperature $\beta>0$ follows the distribution 
\begin{equation} \label{def of Riesz gas Gibbs}
\ud \mathbb{P}_{N,\beta}(X_N) = \frac{1}{Z_{N,\beta}} \exp\Big(-\frac{\beta}{2} N^{-\frac{s}{d}} \mathcal{H}_N(X_N) \Big) \ud X_N. 
\end{equation} 
It is well known that the empirical measure $\frac{1}{N} \sum_{j=1}^N \delta_{x_j}$ converges weakly, as $N\to\infty$, to the unique probability measure $\mu_V$ that minimises the energy functional
\begin{equation} \label{def of energy functional}
I_V[\mu] := \iint_{ (\R^d)^2} g_s(x-y) \ud\mu(x)\ud\mu(y) + \int_{ \R^d } V(x) \ud\mu(x).   
\end{equation}
The minimiser $\mu_V$ is referred to as the equilibrium measure. 

Several important and extensively studied models arise as particular
instances of the Riesz gas. When \(s=d-2\), the interaction reduces to
the classical Coulomb potential, and the corresponding system is known
as the Coulomb gas. When \(s=0\), the interaction becomes logarithmic,
yielding the log gas. In addition, when $d=1$, the log gas coincides with the joint eigenvalue distribution of Hermitian random matrix ensembles 
\cite{Fo10}. On the other hand, when $d=2$, the logarithmic interaction is
equivalent to the Coulomb interaction, and the resulting two-dimensional
Coulomb (or log) gas is closely connected with non-Hermitian random
matrix ensembles, most notably the complex Ginibre ensemble
\cite{BF25}. We also refer the reader to the recent review~\cite{BF25a} for an electrostatic perspective on random matrix theory.

\smallskip 
 
Owing to its simple formulation and its ability to interpolate
continuously between different interaction regimes, the Riesz gas
provides a natural framework for investigating how the interaction
range shapes the collective behaviour of many-particle systems. This
question has been investigated extensively through large-deviation
principles, fluctuation theory, and other probabilistic aspects of the model; see, e.g., \cite{Se24,HLSS18,LS17,LS25,LS25a} and references therein. 

We approach this problem from a geometric
viewpoint by investigating the structure of constrained equilibrium
configurations. 
Our aim is to understand how the geometry of equilibrium states depends on the interaction exponent~$s$, and whether it exhibits qualitative transitions that distinguish different long-range interaction regimes beyond the classical long-range/short-range dichotomy.  
To this end, we investigate the equilibrium measure of a Riesz gas confined to
a half-space. More precisely, our main question is the following:
\begin{center}
\emph{How does the equilibrium state of a Riesz gas change when it is confined to a half-space?}
\end{center} 
For logarithmic interactions in
dimensions $d=1$ and $d=2$, this problem admits a natural probabilistic
interpretation in terms of large-deviation events for extremal
eigenvalues of random matrices. 
For $d=1$, the constrained equilibrium problem and its relation
to extreme eigenvalue statistics have been studied extensively; see, e.g., 
\cite{DM06,DM08,MNSV09,BDG01,KC10,RKC12,MV09,FW12,MS14,MNSV11}. For $d=2$, the problem is closely related to the distribution of the
largest real part of the eigenvalues of the Ginibre ensemble. This observable plays a central role in May's celebrated stability criterion for large ecosystems~\cite{May72}; see \cite{BLO26,XZ24} and the references therein.

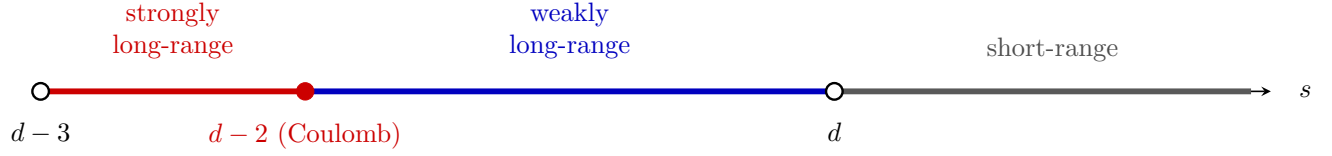
\begin{figure}[t]
    \centering
    \begin{tikzpicture}[
        x=1.75cm,
        y=1cm,
        >=stealth,
        every node/.style={font=\normalsize},
        open point/.style={
            circle,
            draw=black,
            fill=white,
            line width=0.9pt,
            inner sep=2.2pt
        },
        coulomb point/.style={
            circle,
            draw=red!80!black,
            fill=red!80!black,
            inner sep=2.3pt
        }
    ]


    \coordinate (A) at (0,0);    
    \coordinate (B) at (2.0,0);  
    \coordinate (C) at (6.0,0);  
    \coordinate (D) at (9.3,0);  


    \draw[->,semithick]
        (A) -- (D)
        node[right=7pt] {$s$};


    \draw[
        red!80!black,
        line width=2.2pt
    ]
        ($(A)+(0.05,0)$) -- (B);

    \draw[
        blue!75!black,
        line width=2.2pt
    ]
        ($(B)+(0.05,0)$) -- ($(C)+(-0.05,0)$);

    \draw[
        black!65,
        line width=2.2pt
    ]
        ($(C)+(0.05,0)$) -- ($(D)+(-0.15,0)$);


    \node[open point] at (A) {};
    \node[coulomb point] at (B) {};
    \node[open point] at (C) {};


    \node[
        text=red!80!black,
        align=center,
    ]
        at ($(A)!0.5!(B)+(0,0.80)$)
        {strongly\\long-range};

    \node[
        text=blue!75!black,
        align=center,
    ]
        at ($(B)!0.5!(C)+(0,0.80)$)
        {weakly\\long-range};

    \node[
        text=black!70,
        align=center,
    ]
        at ($(C)!0.5!(D)+(0,0.55)$)
        {short-range};


    \node[below=8pt]
        at (A)
        {$d-3$};

    \node[
        below=8pt,
        text=red!80!black
    ]
        at (B)
        {$d-2$ (Coulomb)};

    \node[below=8pt]
        at (C)
        {$d$};

    \end{tikzpicture}

    \caption{
        Interaction regimes for the Riesz exponent \(s\).
    }
    \label{fig:interaction-regimes}
\end{figure}

The main contribution of the present paper is the analysis of a sharp dichotomy within the long-ranged regime. Our main findings can be summarised as follows. 
\begin{itemize}
    \item In the \textbf{weakly long-ranged regime}, where \(s\in(d-2,d)\), the Riesz gas never becomes completely confined to the hard wall, no matter how far the gas is pushed; see Figure~\ref{Fig_weakly}.
    
    \smallskip

    \item In the \textbf{strongly long-ranged regime}, where \(s\in(d-3,d-2]\), there exists a 
    critical value \(a_{\rm cri}\) 
    beyond which the Riesz gas becomes completely confined to the hard wall; see Figure~\ref{Fig_strongly}.
\end{itemize}
This extends recent results for the Coulomb case \cite{DKMSS17,ASZ14,BFMS26,FIT26}. We discuss our results more precisely in the next section.

\begin{figure}[t]
\centering
\setlength{\tabcolsep}{3pt}
\renewcommand{\arraystretch}{1.1}

\begin{tabular}{
|>{\centering\arraybackslash}m{0.30\textwidth}
|>{\centering\arraybackslash}m{0.30\textwidth}
|>{\centering\arraybackslash}m{0.34\textwidth}|
}
\hline
\cellcolor{blue!3}
$a=0$
&
\cellcolor{blue!3}
$a=1$
&
\cellcolor{blue!3}
$a=2$
\\
\hline

\includegraphics[
width=\linewidth,
height=4.8cm,
keepaspectratio
]{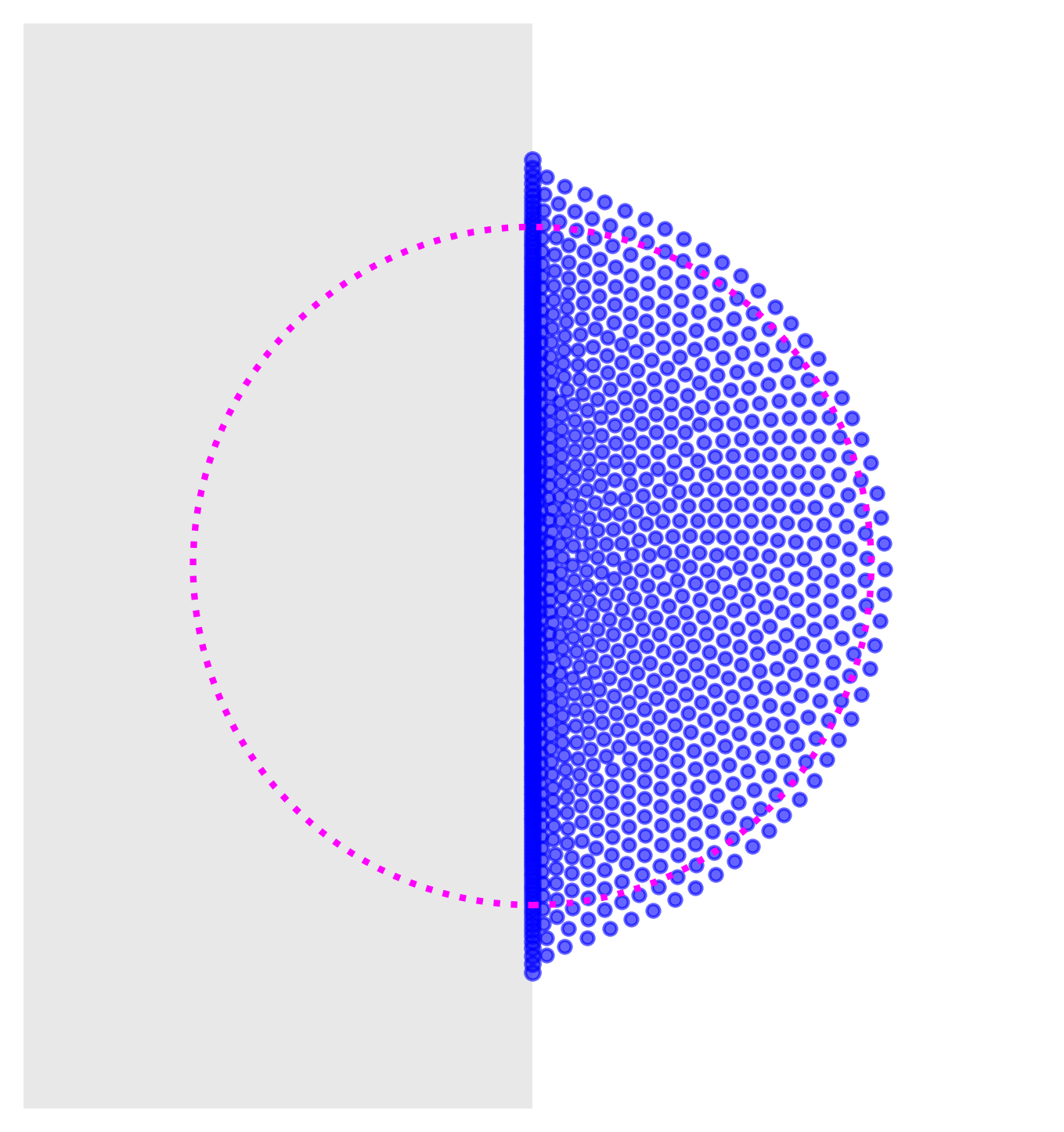}
&
\includegraphics[
width=\linewidth,
height=4.8cm,
keepaspectratio
]{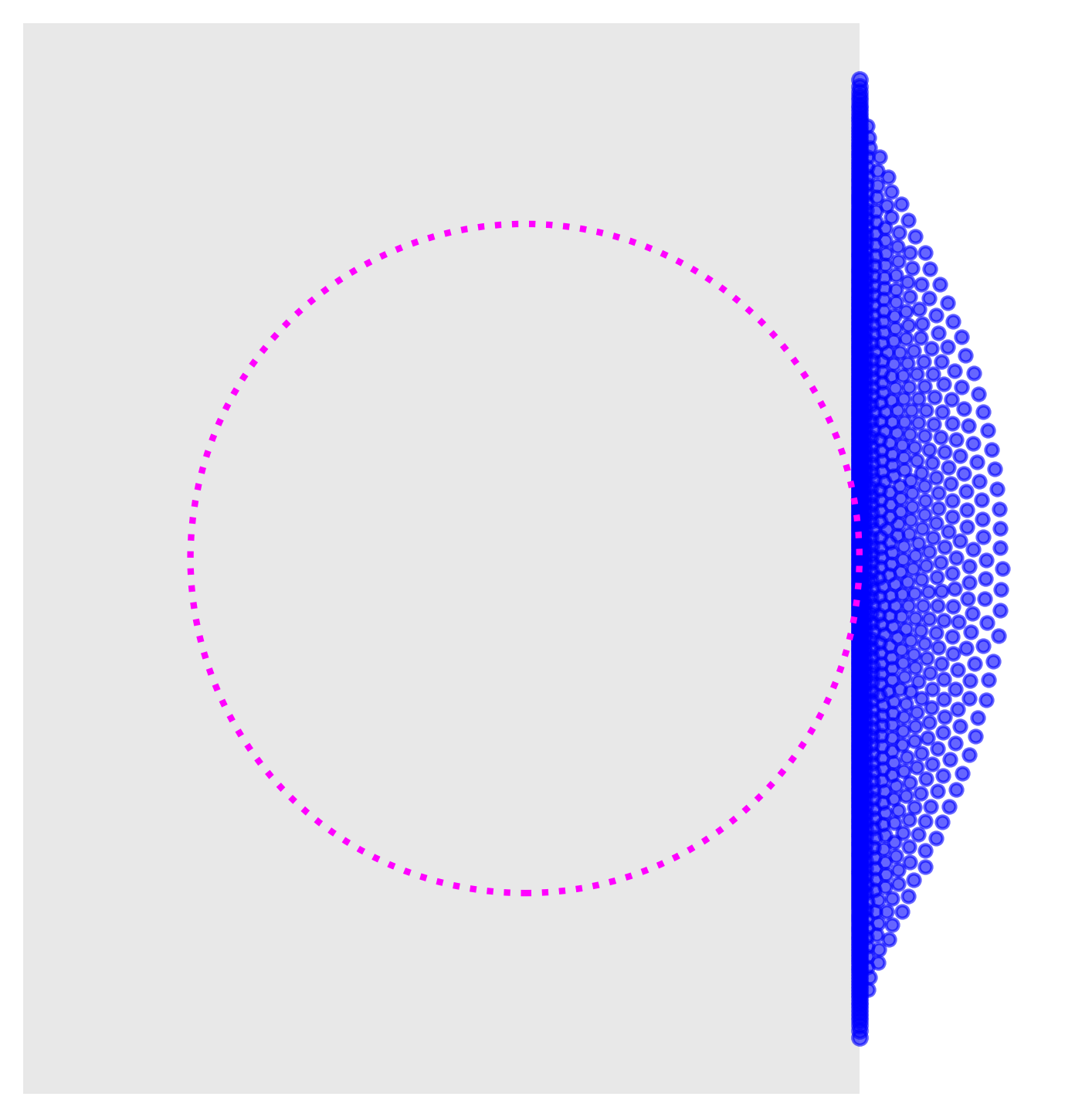}
&
\includegraphics[
width=\linewidth,
height=4.8cm,
keepaspectratio
]{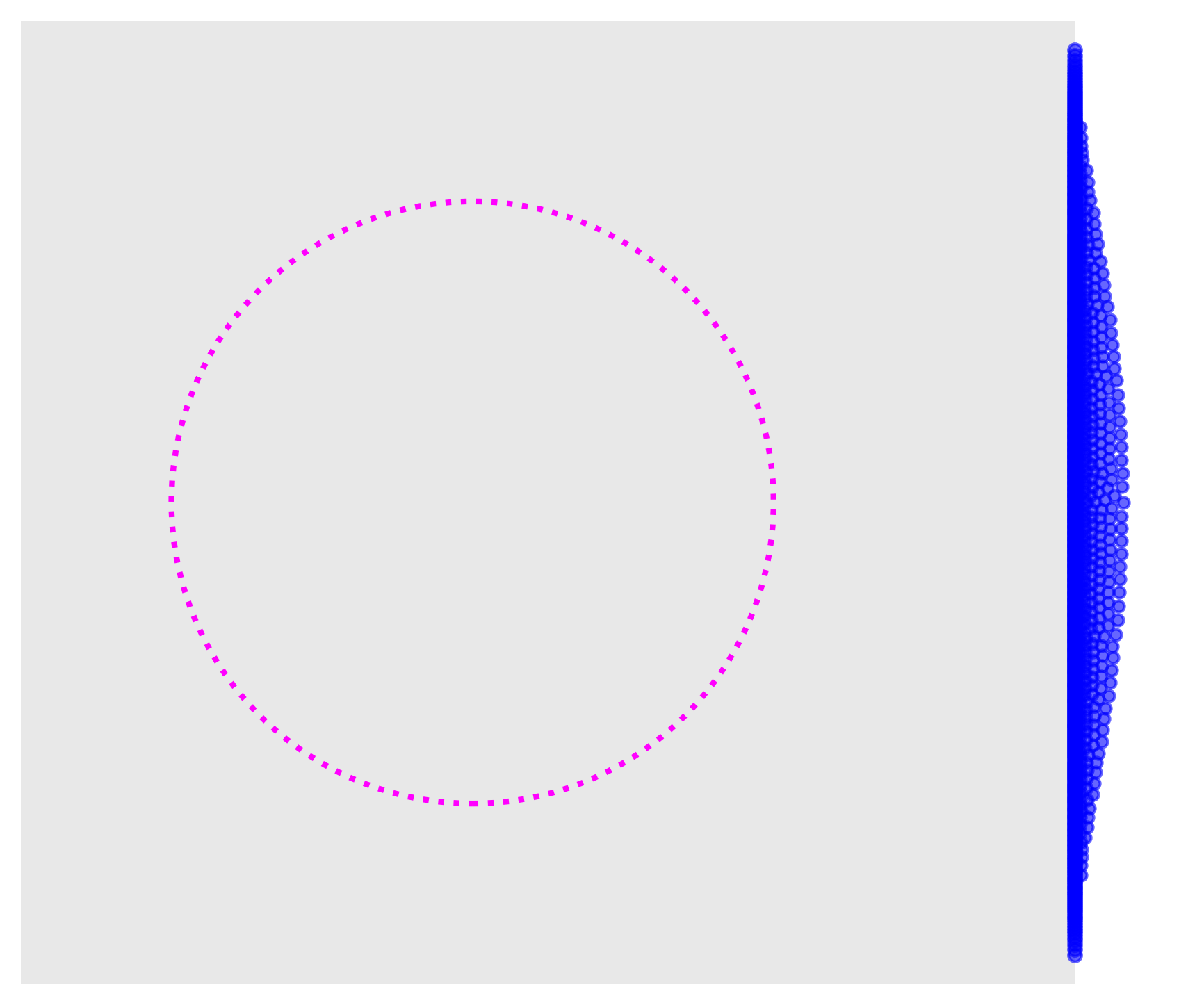}
\\
\hline
\end{tabular}

\caption{Point configurations of the minimiser of \eqref{def of Hamiltonian}, obtained numerically via gradient descent for $N=1000$. Here, $d=2$ and $s=\frac{1}{2}$, corresponding to the weakly long-ranged regime. The shaded region represents the hard wall, while the dashed circle indicates the boundary of the unconstrained equilibrium measure. 
} 
\label{Fig_weakly}
\end{figure}

\begin{figure}[t]
\centering
\setlength{\tabcolsep}{4pt}
\renewcommand{\arraystretch}{1.1}

\begin{tabular}{
|>{\centering\arraybackslash}m{1.5cm}
|>{\centering\arraybackslash}m{0.28\textwidth}
|>{\centering\arraybackslash}m{0.28\textwidth}
|>{\centering\arraybackslash}m{0.28\textwidth}|
}
\hline
\cellcolor{blue!6}
&
\multicolumn{2}{c|}{\cellcolor{blue!3}\(a<a_{\rm cri}\)}
&
\cellcolor{blue!3}\(a=a_{\rm cri}\)
\\
\hline

\cellcolor{blue!3}
\(\displaystyle s=0\)
&
\includegraphics[width=\linewidth,height=4.8cm,keepaspectratio]{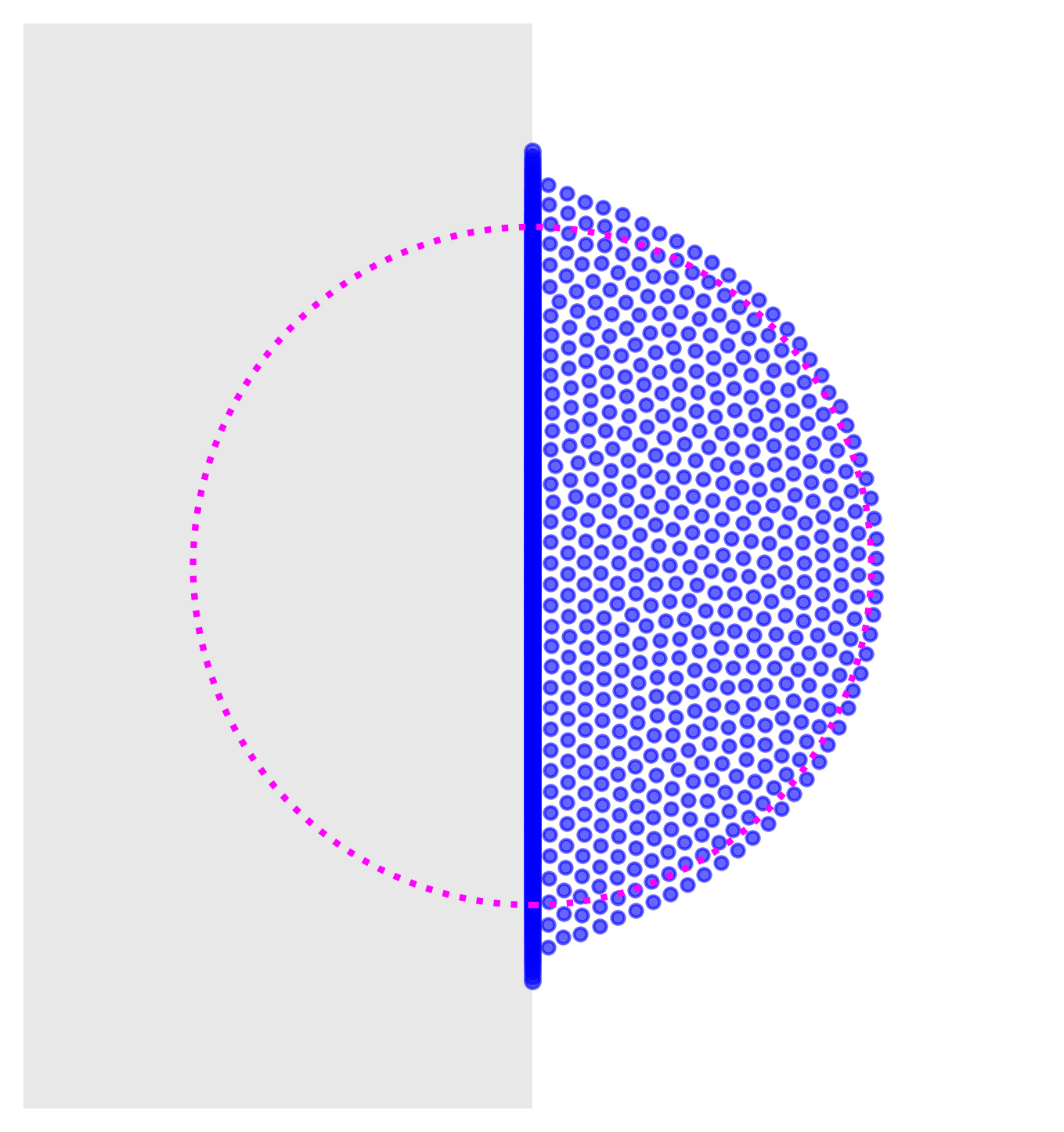}
&
\includegraphics[width=\linewidth,height=4.8cm,keepaspectratio]{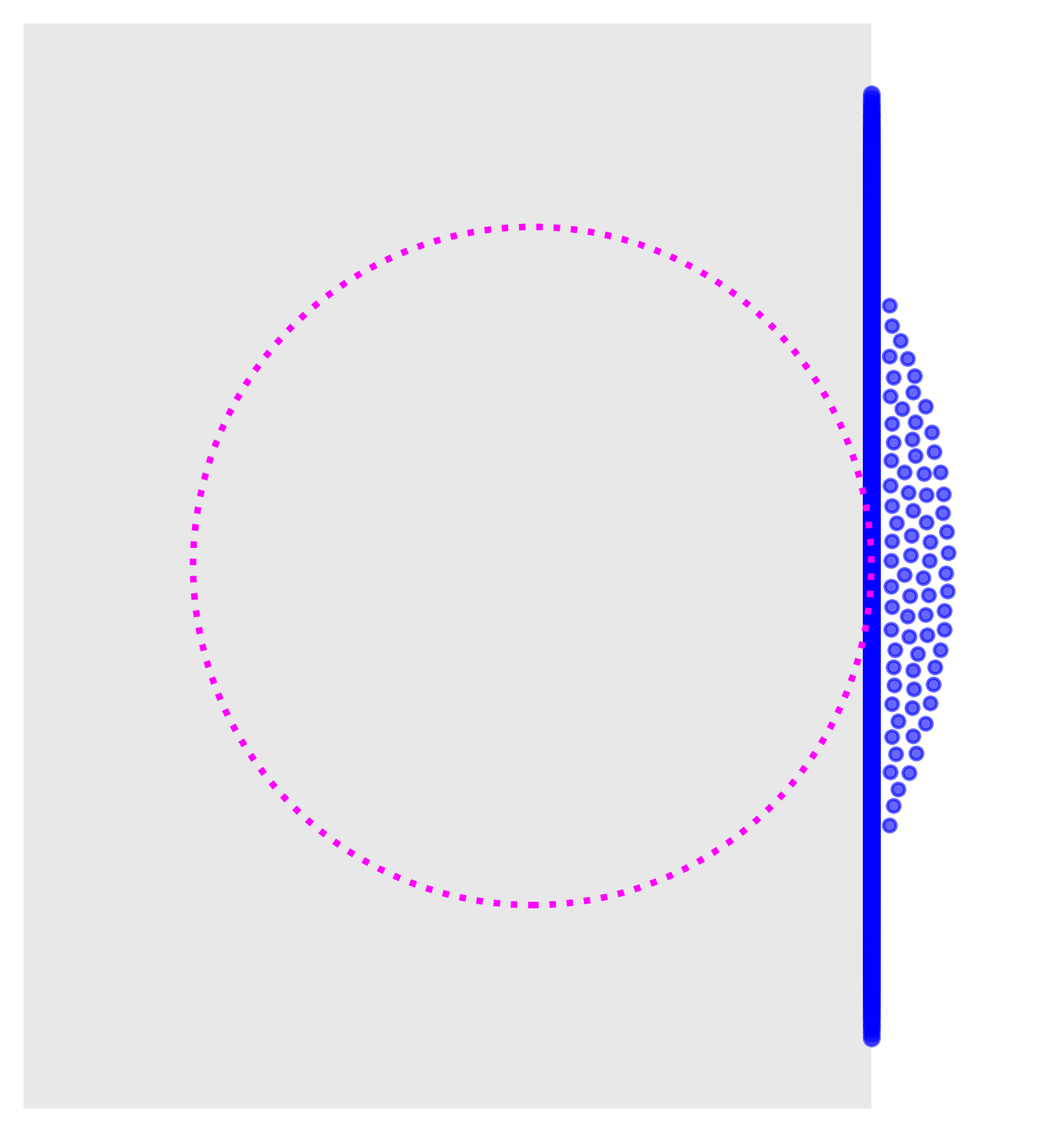}
&
\includegraphics[width=\linewidth,height=4.8cm,keepaspectratio]{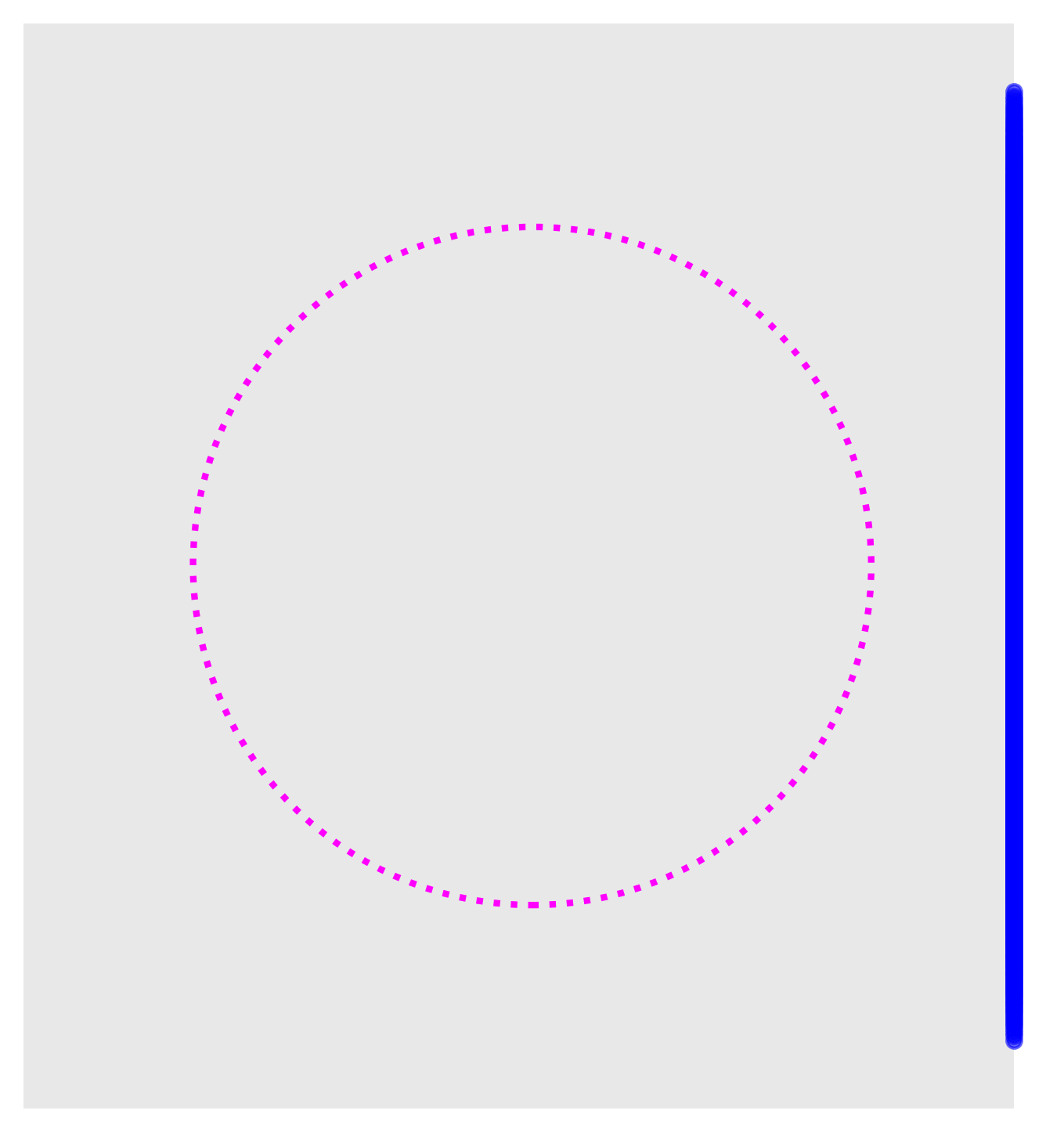}
\\
\hline

\cellcolor{blue!3}
$s=-\frac35$
&
\includegraphics[width=\linewidth]{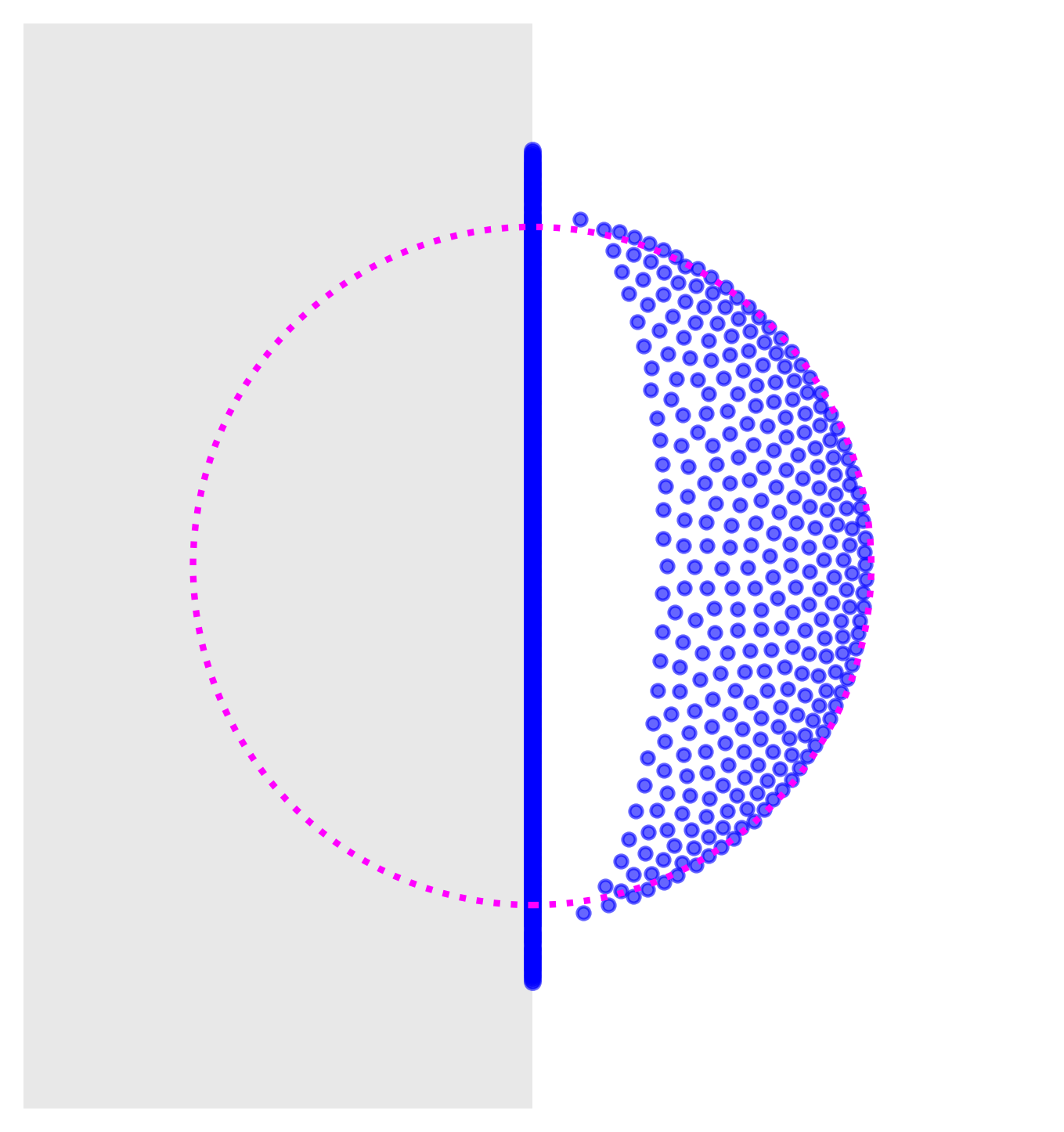}
&
\includegraphics[width=\linewidth]{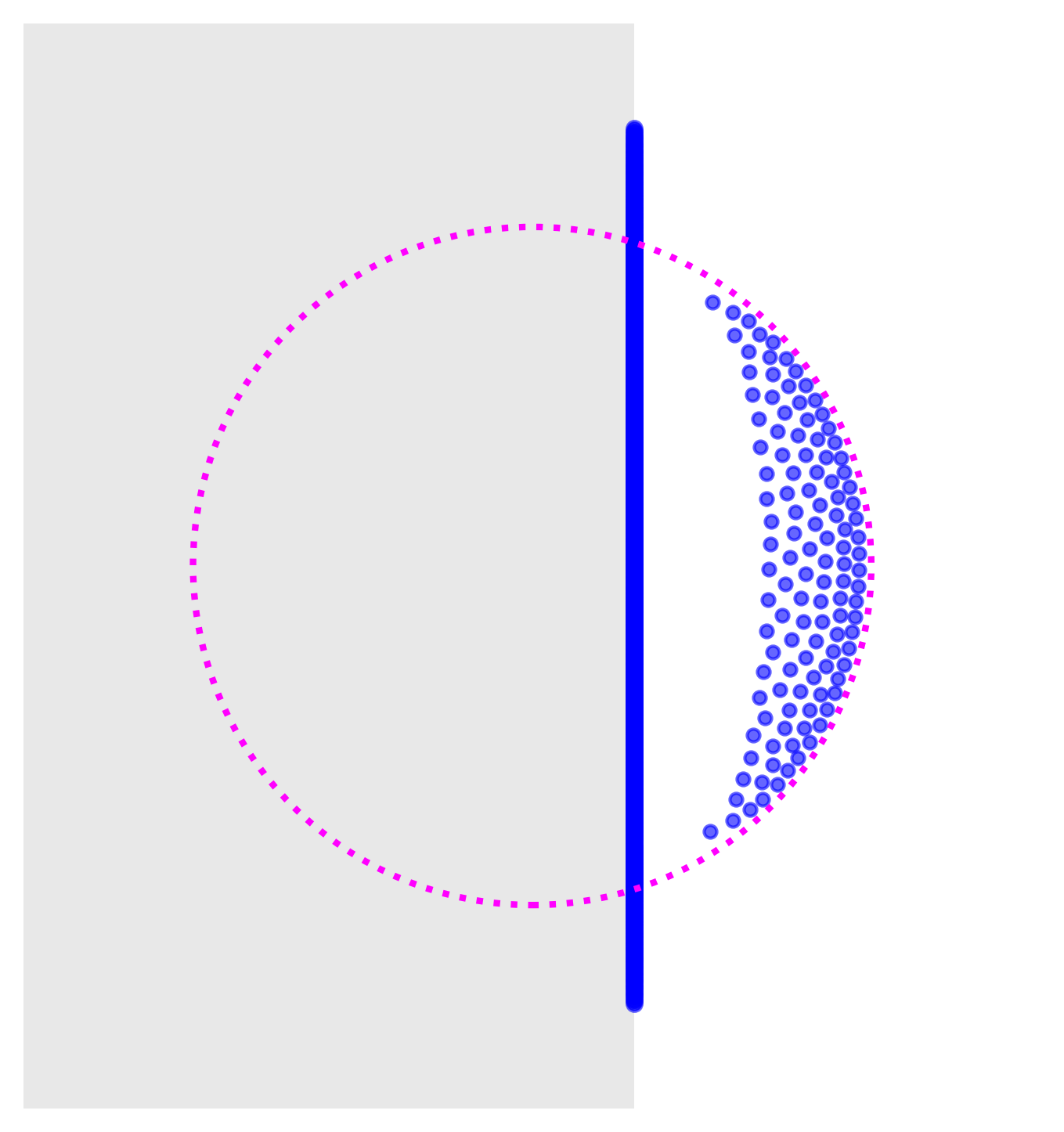}
&
\includegraphics[width=\linewidth]{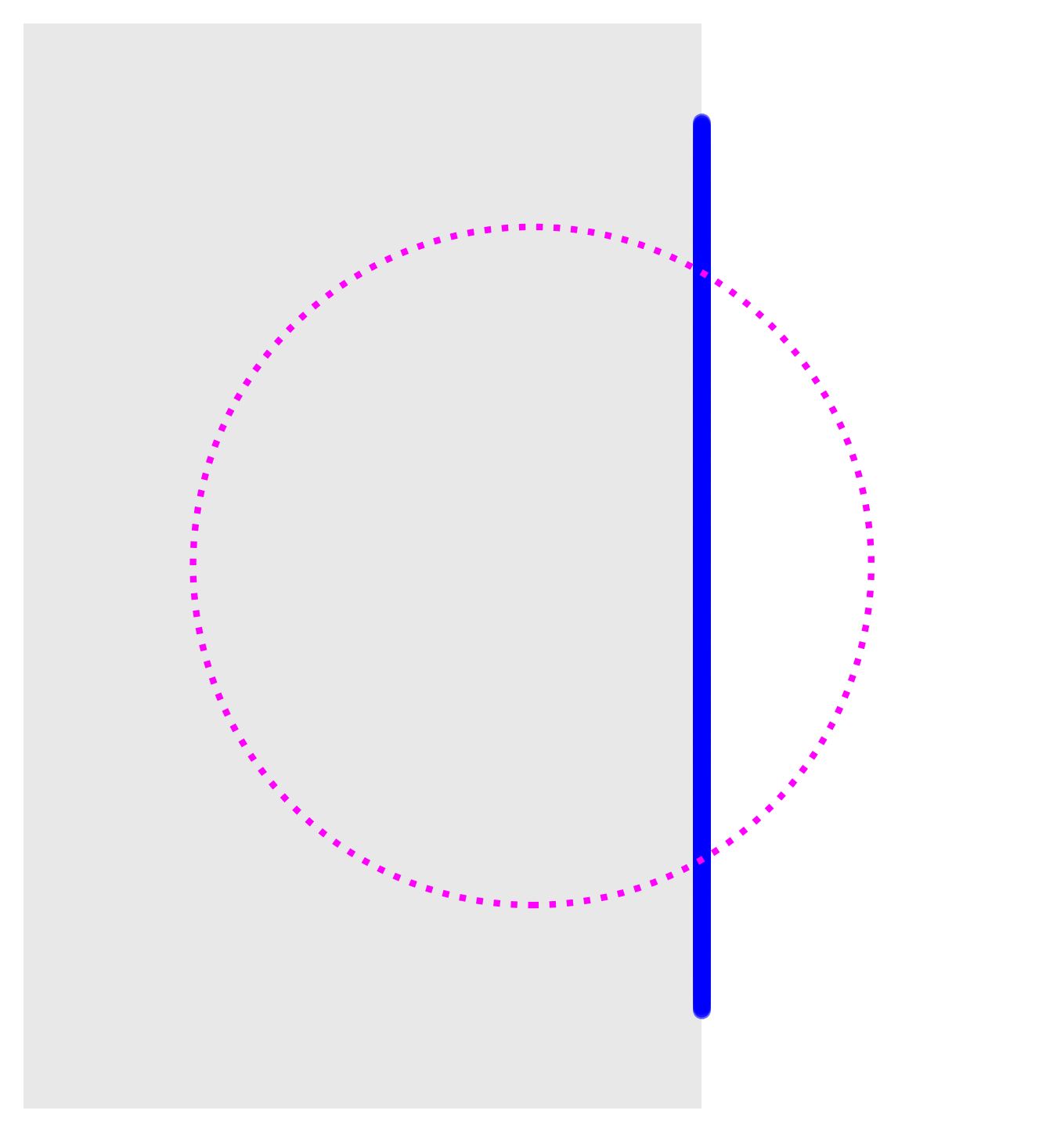}
\\
\hline

\end{tabular}
\caption{Same as Figure~\ref{Fig_weakly}, but in the strongly long-ranged regime, with $s=0$ and $s=-\frac{3}{5}$. As $a$ increases, the particles progressively accumulate on the hard wall, becoming fully confined to it at the critical value $a=a_{\rm cri}(d,s)$. Here, $a_{\rm cri}(d,s)=\sqrt{2}$ for $s=0$, while $a_{\rm cri}(d,s)\approx 0.5$ for $s=-\frac{3}{5}$. 
} \label{Fig_strongly}
\end{figure}

\newpage 

\section{Main results}

Let us now formulate the problem more precisely.
Although the Riesz gas can be defined for all \(s>-2\), throughout this paper we restrict our attention to the range
\begin{equation}
    s\in(d-3,d).
\end{equation}
The restriction \(s<d\) excludes the short-range regime, which does not exhibit
the phenomena that are the focus of the present work. The lower bound
\(s>d-3\) is imposed for a different reason. Indeed, for \(s<d-3\) one expects
the emergence of additional phases, whose analysis requires different
techniques and lies beyond the scope of this paper; see the discussion following Proposition~\ref{Prop_eq msr unconstrained}. We also note that, in the mathematical literature on long-ranged interactions, and in particular on Riesz gases, most existing results (see e.g. \cite{Se24,Fo26}) concern the weakly long-ranged regime \(s\in[d-2,d)\). 

We consider the quadratic confining potential
\begin{equation}\label{def of quad potential no constraint}
  V(x)= \mathsf{C}_{d,s}   |x|^2, \qquad
  \mathsf{C}_{d,s} 
  =
  \frac{\Gamma(\frac{d-s}{2})\Gamma(\frac{s}{2}+2)}
       {\Gamma(\frac d2+1)},
\end{equation}
for which the associated equilibrium measure is supported on the
unit ball. More generally, for an arbitrary constant $\mathsf{C}>0$,
the equilibrium measure is given by the following proposition.

\begin{prop}[\textbf{Unconstrained equilibrium measure associated with quadratic potentials}] \label{Prop_eq msr unconstrained}
Let $d\ge1$, $s\in(d-3,d)$, and let
$V(x)=\mathsf{C}|x|^2$ with $\mathsf{C}>0$.
Then the associated equilibrium measure is given by 
\begin{equation} \label{def of eq msr unconstraint}
    \ud \mu_V(x)= \frac{\Gamma(\frac{s}{2}+2)}{\Gamma(2-\frac{d-s}{2}) }\Big(1- \frac{|x|^2}{ \mathsf{R}^2}\Big)^{ 1-\frac{d-s}{2} } \cdot \mathbbm{1}_{|x|\le \mathsf{R} } \frac{ \ud x }{ (\sqrt{\pi} \, \mathsf{R})^d },   
\end{equation}
where 
\begin{equation} \label{def of radius R unconstraint}
    \mathsf{R} =\bigg(\frac{\Gamma(\frac{d-s}{2})\Gamma(\tfrac{s}{2}+2) }{\Gamma(\frac{d}{2}+1)\mathsf{C}}\bigg)^{\frac{1}{s+2}} . 
\end{equation} 
\end{prop}

In the special case $s=0$, Proposition~\ref{Prop_eq msr unconstrained}
reduces to the classical Wigner semicircle law for $d=1$ and the
circular law for $d=2$. The one-dimensional case of Proposition~\ref{Prop_eq msr unconstrained} was established in \cite{ADKKMMS19}, while the higher-dimensional cases $s=d-3$ and $s\in[d-2,d)$ were obtained in
\cite{CSW22} and \cite{BFMS26}, respectively. The normalisation adopted
in \cite{BFMS26} is chosen so that the equilibrium measure is supported
on the unit ball; see Theorems~2.4 and~2.6 together with Remark~8
therein. For completeness, we include a proof of the remaining case
$s\in(d-3,d-2)$. 

In fact, the restriction \(s>d-3\) is essential for Proposition~\ref{Prop_eq msr unconstrained}. When \(s<d-3\), the interaction becomes even more long-ranged, and the equilibrium measure in the unconstrained setting is expected to differ from \eqref{def of eq msr unconstraint}. In particular, the analysis of \cite{CSW23} in the case \(s=d-4\) suggests that the equilibrium measure develops a singular component supported on the surface of the sphere; see also \cite{CMSVW25}. This gives rise to additional phases, whose characterisation requires a separate analysis. 

\medskip 

We now turn to the case where the Riesz gas is confined to a half-space. To formulate the constrained problem, let \(a\in\mathbb{R}\) and define the external potential by
\begin{equation} \label{def of Va constrained potential}
  V_a(x):= 
    \begin{cases}
        V(x), &x_d\ge a,
        \smallskip 
        \\
        +\infty, & \text{otherwise},
    \end{cases}
\end{equation} 
where $V$ is given by \eqref{def of quad potential no constraint}. 
Note that by \eqref{def of Riesz gas Gibbs}, the infinite potential on the region $\{x_d<a\}$ forces all particles to remain in the half-space $\{x_d\ge a\}$. (Note that this situation is fundamentally different from the
well-studied balayage setting arising in the study of hole
probabilities and related statistics
\cite{AR17,Ad18,Ch23a,Ch23,BP26,CW25,CFLV17,CFLV19,AZ15}, where the hard wall
lies strictly inside the support of the equilibrium measure.)

In the sequel, we use notation $x= (x_1,\ldots, x_{d-1},x_d) =(\hat{x}, x_d)$. 
We denote by $\wh \mu_a \equiv \mu_{ V_a }$ the associated equilibrium measure.  
Due to this half-space constraint, the equilibrium is in general not absolutely continuous with respect to the Lebesgue measure in $\R^d$. Consequently, we write 
\begin{equation} 
    \ud\wh{\mu}_a(x) = \ud \nu_a(\hat{x})\otimes \delta_a(x_d) + \ud\mu_a(x). 
\end{equation}

For Coulomb gases ($s=d-2$), constrained equilibrium measures $\wh{\mu}_a$ have been the subject of several recent studies.  One of the most interesting phenomena, first investigated for dimensions
$d=1,2$ in \cite{DKMSS17,ASZ14}, is the existence of a critical wall
location beyond which the equilibrium measure becomes completely confined to
the boundary. More precisely, when the wall is placed beyond the
critical value, the singular component $\nu_a$ carries the entire mass.
This phenomenon was subsequently investigated in arbitrary dimensions
in \cite{BFMS26}, where it was proposed that the critical threshold is
given by \eqref{def of critical value Coulomb} below.   
To be specific, \cite{BFMS26} established one direction of the
Euler--Lagrange characterisation (and a complete proof in dimension
$d=4$). Very recently, the remaining direction was proved in
\cite{FIT26}, thereby confirming that this complete confinement
phenomenon indeed occurs for Coulomb gases in every dimension.

While the works \cite{DKMSS17,ASZ14,BFMS26,FIT26} provide a comprehensive understanding of the complete confinement transition for Coulomb gases, they leave open a more fundamental question: is complete confinement a universal feature of long-range interactions, or is it a genuinely interaction-dependent phenomenon? Since the Coulomb interaction corresponds to the single value \(s=d-2\) within the Riesz family, the existing results alone cannot answer this question. Indeed, the extensively studied one-dimensional log gas exhibits no complete confinement transition, suggesting that qualitatively different behaviour may arise beyond the Coulomb case. 

Our main result reveals that the interaction exponent $s=d-2$ marks a
fundamental threshold within the long-range regime. More precisely, the
geometry of the constrained equilibrium measure exhibits a sharp dichotomy; cf. Figures~\ref{fig:interaction-regimes}, ~\ref{Fig_weakly} and ~\ref{Fig_strongly}. 

To describe the phase transition, we first introduce the critical value.
For this purpose, we recall that the incomplete beta function
\cite[Section~8.17]{NIST} is defined by
\begin{equation} \label{def of incomplete beta}
\mathrm{B}_z(p,q)
:=
\int_0^z
u^{p-1}(1-u)^{q-1}\,\ud u,
\qquad
(\re p>0)
\end{equation} 
and extended to general parameters by analytic continuation. 

\begin{defn}{Critical value} Let $d \ge 1$ and $s \in (d-3,d-2)$. 
Let \(\mathsf{x}\) be the unique nonnegative solution of
\begin{equation}
\label{eq:critical point incomplete beta}
\mathsf{x}^{-2\alpha}
    (1+\mathsf{x}^2)^{\frac{3-d}{2}}
    +
    \Big(
        \frac{s}{2}
        -
        \frac{\alpha}{\mathsf{x}^2}
    \Big)
    \mathrm{B}_{\frac{\mathsf{x}^2}{1+\mathsf{x}^2}}
    \Big(
        \frac{d-s-1}{2},
        \frac{s}{2}
    \Big)
    =
    \mathsf{C}_{d-1,s},
\end{equation}
where $\mathrm{B}$ is the incomplete beta function \eqref{def of incomplete beta} and 
\begin{equation} \label{def of alpha}
 \alpha := \frac{s-d+3}{2}. 
\end{equation} 
Then the critical value is defined by \begin{equation}
\label{eq:a critical incomplete beta}
    a_{\rm cri}(d,s)
    :=
 \frac{1}{R^{s+1}}    \frac{
        \Gamma(\frac{d}{2}+1)
    }{
        \Gamma(\frac{d-s}{2})
        \Gamma(\frac{s}{2}+1)
    }  \, 
    \frac{1}{\mathsf{x}}
    \mathrm{B}_{\frac{\mathsf{x}^2}{1+\mathsf{x}^2}}
    \Big(
        \frac{d-s-1}{2},
        \frac{s}{2}
    \Big),
\end{equation}
where \begin{equation} \label{def of R}
    R
    :=\Big(\frac{\mathsf{C}_{d-1,s}}{\mathsf{C}_{d,s}}\Big)^{\frac{1}{s+2}}=
    \bigg(
    \frac{
    \Gamma(\tfrac{d-s-1}{2})
    \Gamma(\tfrac{d}{2}+1)
    }{
    \Gamma(\tfrac{d-s}{2})
    \Gamma(\tfrac{d+1}{2})
    }
    \bigg)^{\frac{1}{s+2}}.
\end{equation}
\end{defn}

\begin{figure}[t]
    \centering
    \includegraphics[width=0.6\linewidth]{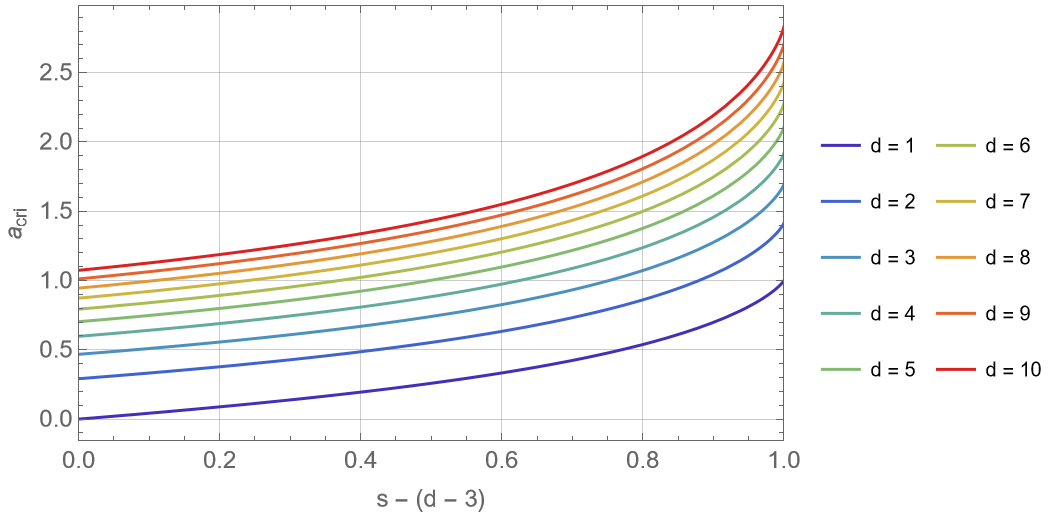}
  \caption{Graphs of the function $s\mapsto a_{\rm cri}(d,s)$ for $d=1,\dots,10$. The horizontal axis represents $s-(d-3)$, which ranges from $0$ to $1$ for every $d$.}
    \label{fig: a-cri}
\end{figure}

See Figure~\ref{fig: a-cri} for the graph of $a_{\rm cri}(d,s)$.
One observes that $s \mapsto a_{\rm cri}(d,s)$ is increasing.
This monotonicity is intuitively natural: as the interaction becomes more
long-ranged, the influence of the hard wall propagates over larger
distances, making complete confinement possible even when the wall is
located farther away. See Remark~\ref{Rem_support bdy} for further
discussion. 

The existence and uniqueness of the nonnegative solution to
\eqref{eq:critical point incomplete beta} will be established in
Lemma~\ref{lem: property f(x)}.  
We note that the extremal case $s=d-2$ requires separate treatment, for which we define $\mathsf{x}=0$. The corresponding critical value is then obtained by continuity as the limit $s\uparrow d-2$, yielding the explicit formula \eqref{def of critical value Coulomb}; see Remark~\ref{Rem_critical special} for further details.

We are now ready to state our main result.

\begin{thm}[\textbf{Interaction-range dichotomy for full confinement}] \label{thm: fully confined}
Let $d\geq 1$ and $s\in(d-3,d)$. For $a \in \R$, let $\wh \mu_a $ be the equilibrium measure associated with the potential \eqref{def of Va constrained potential}.  

\begin{enumerate}[label=\textup{(\roman*)}]
\item \label{thm: fully confined weak}
\textup{\textbf{Weakly long-ranged regime:}}
Suppose that $s\in(d-2,d)$. Then $\widehat{\mu}_a$ is never fully confined to the hard wall (i.e., $ |\nu_a|<1$ for all $a \in \R$). 

\smallskip

\item
\textup{\textbf{Strongly long-ranged regime:}} \label{thm: fully confined strong}
Suppose that $s\in(d-3,d-2]$. Then $\widehat{\mu}_a$ is fully confined to the hyperplane $\{x_d=a\}$ (i.e., $  |\nu_a|=1$) if and only if $a\ge a_{\mathrm{cri}}(d,s)$.  
Furthermore, in this case, we have 
\begin{equation}
    \ud\widehat{\mu}_a(x)
    =
    \ud \nu_a(\hat{x})\otimes \delta_a(x_d),
\end{equation}
where
\begin{equation} \label{def of fully singular msr}
   \ud\nu_a(\hat{x})
    =  \frac{
    \Gamma(2+\tfrac{s}{2})
    }{
    \Gamma(\tfrac{5+s-d}{2})
    }
    \Big(
    1-\frac{|\hat{x}|^2}{R^2}
    \Big)^{ \alpha }
    \mathbbm{1}_{|\hat{x}|\le R}
    \frac{\ud\hat{x}}
    {(\sqrt{\pi}R)^{d-1}}.  
\end{equation}
Here  $\alpha$ and $R$ are given by \eqref{def of alpha} and \eqref{def of R}, respectively.  
\end{enumerate}
\end{thm}

As previously mentioned, the numerical results illustrating Theorem~\ref{thm: fully confined} (i) are presented in Figure~\ref{Fig_weakly}, while those corresponding to Theorem~\ref{thm: fully confined} (ii) are shown in Figure~\ref{Fig_strongly}.

We conclude this section with several remarks concerning Theorem~\ref{thm: fully confined}. 

\begin{rem}[Coulomb case]
Note that, in the Coulomb case $s=d-2$, Theorem~\ref{thm: fully confined} recovers the previously established results of
\cite{DKMSS17,ASZ14,BFMS26,FIT26}.
Indeed, the critical value \eqref{eq:a critical incomplete beta} simplifies to\footnote{Our convention differs slightly from that of
\cite{BFMS26}. There, a $(d+1)$-dimensional Coulomb gas is confined to a
$d$-dimensional hyperplane, whereas here we consider a
$d$-dimensional Riesz gas confined to a $(d-1)$-dimensional
hyperplane.}
\begin{equation} \label{def of critical value Coulomb}
 a_{\rm cri}(d,d-2)
 = d \bigg( \frac{ \Gamma(\frac{d+1}{2})  }{ \sqrt{\pi} \, \Gamma(\frac{d}{2}+1)   } \bigg)^{ \frac{d-1}{d} }.
\end{equation}
This is precisely the critical value introduced in \cite{BFMS26}.
Furthermore, the equilibrium measure \eqref{def of fully singular msr} reduces to that given in \cite[Eq.~(2.54)]{BFMS26}. Consequently,
\cite[Theorem~2.7]{BFMS26} and \cite[Theorem~1.2 (c)]{FIT26} are recovered as special cases of Theorem~\ref{thm: fully confined}.
\end{rem}

\begin{rem}[One-dimensional Riesz case; metastable regime] The equilibrium measure for the one-dimensional Riesz gas under a half-space constraint was studied in detail in \cite{KKKMMS21}; see also \cite{KKKMMS22,KKKMMS24}. 

The authors of \cite{KKKMMS21} adopt a different normalisation, considering the quadratic potential $x^2/|s|$. This differs from our convention \eqref{def of quad potential no constraint},
and consequently, the support of the unconstrained equilibrium measure in \cite{KKKMMS21} has a different radius.
The authors then introduced an analytic threshold
\cite[Eq.~(18)]{KKKMMS21}\footnote{More precisely, the formula in \cite[Eq.~(18)]{KKKMMS21} differs by an overall minus sign, since the authors consider a hard-wall constraint on the right half-line, whereas our convention imposes the constraint on the left half-line.}
\begin{equation}
    w_c(s):=\frac{(s+2)\,|s(s+1)|^{\frac{1}{s+2}}}{|s+1|},
\end{equation}
and showed that, whenever the wall position $w$ satisfies $w>w_c(s)$, the constrained equilibrium measure is fully confined to the hard wall, namely, it reduces to a Dirac mass. 
Furthermore, based on numerical simulations, the authors conjectured the existence of a genuine critical value $w^*(s)<w_c(s)$ such that complete confinement occurs if and only if $w>w^*(s)$.

Our results confirm this conjecture and, moreover, provide an explicit analytic expression for $w^*(s)$. To see this, we first observe that, in the one-dimensional case $d=1$, the critical value \eqref{eq:a critical incomplete beta} simplifies to
\begin{equation}  \label{def of critical value d=1}
a_{\rm cri}(1,s) = \frac{s+2}{2|s+1|}\bigg(\frac{s+1}{s}\frac{\sqrt{\pi}}{\Gamma(\frac{1-s}{2})\Gamma(\frac{s}{2}+2)}\bigg)^{\frac{1}{s+2}}. 
\end{equation}   
After an appropriate rescaling, we obtain, for $s\in(-2,-1)$, 
\begin{equation}\label{eq:true w^*(s)}
    w^*(s) = (|s|\mathsf{C}_{1,s})^{\frac{1}{s+2}}a_{\rm cri}(1,s) = (s+2)(2|s+1|)^{-\frac{s+1}{s+2}}.  
\end{equation}
In particular, one readily verifies that $w^*(s) < w_c(s)$, thereby confirming the conjectured strict inequality.  

In fact, these values admit a clear analytic interpretation. In \cite{KKKMMS21}, the value $w_c(s)$ arises from the saddle-point analysis, or equivalently, from the Euler--Lagrange variational equality. More precisely, $w_c(s)$ is the threshold beyond which a measure with both singular and absolutely continuous components can no longer satisfy the Euler--Lagrange equality. However, the Euler--Lagrange equality alone is not sufficient to characterise the equilibrium measure; one must also verify the corresponding Euler--Lagrange variational inequality. As we prove in the present paper, the value $w^*(s)$ is precisely the unique threshold at which this variational inequality is first satisfied. Consequently, $w^*(s)$, rather than $w_c(s)$, is the true critical value for the confinement transition.
(We note that an analogous phenomenon also arises in the two-dimensional Coulomb gas, albeit for a different reason, namely a change in the topology of the candidate equilibrium measures; see \cite{BY26}.) 

The authors of \cite{KKKMMS21} interpreted the regime $w^*(s)<w<w_c(s)$ as \emph{metastable}, since both the Dirac mass on the hard wall and a mixed measure satisfy the Euler--Lagrange equality and therefore correspond to stationary points of the weighted Riesz energy functional. This also explains the numerical instability in approximating the equilibrium measure. 
For instance, when $s=-1.5$, the exact formula \eqref{eq:true w^*(s)} yields $w^*(s)=0.5$, while the numerical estimate reported in \cite{KKKMMS21} is approximately $0.441$. 
\end{rem} 

\begin{rem}[Complete confinement before the support boundary] \label{Rem_support bdy}
In the strongly long-ranged regime, where a complete confinement transition occurs, several additional interesting features emerge.

We first consider the Coulomb case $s=d-2$. One can notice that the critical value \eqref{def of critical value Coulomb} is increasing in the spatial dimension \(d\). Moreover, as discussed in \cite[Remark~12]{BFMS26}, it diverges as $d\to\infty$. From the perspective of statistical physics, this monotonicity is quite natural. In higher dimensions, particles have greater freedom to spread out, so a stronger hard-wall effect is required to confine the entire gas.

Another noteworthy property is that $a_{\rm cri}(d,d-2)\ge1,$ with equality attained only in the one-dimensional case $d=1$. Since our normalisation is chosen so that the unconstrained equilibrium measure is supported on the unit ball (see \eqref{def of eq msr unconstraint}), this inequality implies that, in the Coulomb case, the hard wall must be pushed beyond the boundary of the unconstrained support before complete confinement can occur. For example, in the first row of Figure~\ref{Fig_strongly}, corresponding to $d=2$, the critical value is $\sqrt{2}$, meaning that the wall must be moved past the boundary of the circular law before all particles become confined.

On the other hand, for interactions that are more long-ranged than the Coulomb interaction, the effect of the hard wall is expected to be stronger. Consequently, complete confinement may occur even when the wall is located strictly inside the support of the unconstrained equilibrium measure. For example, the second row of Figure~\ref{Fig_strongly} illustrates this phenomenon: one can see that the particles become fully confined even before the hard wall reaches the rightmost point of the disc. 
Motivated by this observation, for each fixed dimension $d$, one may define the value $s_*=s_*(d)$ by
\begin{equation}
a_{\rm cri}(d, s_*)=1.
\end{equation}
Then, for $s<s_*(d)$, the interaction is sufficiently long-ranged that complete confinement occurs before the wall reaches the boundary of the unconstrained support. Our numerical computations suggest that such a threshold exists for $d\le8$; see Table~\ref{tab:s-cri}. (Recall that the quantity $s_*-(d-3)$ necessarily belongs to the interval $(0,1]$.) Note that this threshold corresponds to the horizontal cross-section of the surface in Figure~\ref{fig: a-cri} at the level $a_{\rm cri}=1$.

On the other hand, in sufficiently high dimensions, the particles have enough freedom to spread that, throughout the strongly long-ranged regime $s\in(d-3,d-2]$, one always needs to move the wall beyond the radius of the unconstrained equilibrium support in order to achieve complete confinement.
\end{rem}

 \begin{table}[t]
    \centering
    \renewcommand{\arraystretch}{1.5}
    \begin{tabular}{K{1.8cm} K{1.2cm} K{1.2cm} K{1.2cm} K{1.2cm} K{1.2cm} K{1.2cm} K{1.2cm} K{1.2cm}}
        \hline
        \rowcolor{gray!10}
        $d$
            & $1$ & $2$ & $3$ & $4$ & $5$ & $6$ & $7$ & $8$ \\
        \hline
        \cellcolor{gray!10}$s_*$
            & $-1.0000$
            & $-0.1204$
            & $0.7520$
            & $1.6249$
            & $2.4974$
            & $3.3688$
            & $4.2390$
            & $5.1078$ \\
        \hline
        \cellcolor{gray!10}$s_*-(d-3)$
            & $1.0000$
            & $0.8796$
            & $0.7520$
            & $0.6249$
            & $0.4974$
            & $0.3688$
            & $0.2390$
            & $0.1078$ \\
        \hline
    \end{tabular}
    \caption{The exponent \(s_*\) satisfying \(a_{\rm cri}(d,s_*)=1\).}
    \label{tab:s-cri}
\end{table}

\subsection*{Plan of the paper} The rest of this paper is organised as follows. In Section~\ref{Section_prelim}, we review the necessary preliminaries, establish several auxiliary lemmas, and prove Proposition~\ref{Prop_eq msr unconstrained}. Section~\ref{sec: proof of main thm} is devoted to the proof of Theorem~\ref{thm: fully confined}.

\subsection*{Acknowledgements} The authors were supported by the National Research Foundation of Korea grant (RS-2025-00516909, RS-2026-25518141). We thank Peter J. Forrester, Satya N. Majumdar, and Gregory Schehr for their interest and helpful discussions during the initiation and preparation of this paper.

\section{Preliminaries} \label{Section_prelim}

In this section, we prove Proposition~\ref{Prop_eq msr unconstrained} and establish several auxiliary lemmas that will be used in the next section.

We first recall the Euler--Lagrange variational conditions characterising the equilibrium measure. By Frostman's theorem~\cite{Fr35} (see also \cite[Theorem~2]{DOSW23}), the equilibrium measure $\mu_V$ satisfies 
\begin{equation}
\label{eq:EL condition}
    2\int_{\R^d} g_s(x-y) \ud \mu_V(x) +V(y) \begin{cases}
        = c, & y \in \operatorname{supp} \mu_V,
        \smallskip 
        \\
        \ge c, &y \notin \operatorname{supp}\mu_V,
    \end{cases}
\end{equation}
where the constant $c$ is referred to as the (modified) Robin constant. Conversely, once the existence of a minimiser of the Riesz energy functional~\eqref{def of energy functional} is established, any probability measure satisfying~\eqref{eq:EL condition} necessarily coincides with the unique equilibrium measure $\mu_V$. 

We present the proof of Proposition~\ref{Prop_eq msr unconstrained}.

\begin{proof}[Proof of Proposition~\ref{Prop_eq msr unconstrained}]
As noted after the statement of Proposition~\ref{Prop_eq msr unconstrained}, it suffices to establish the case $d \ge 2$ and $s\in(d-3,d-2)$, which does not appear to be covered in the existing literature. In fact, the proof of~\cite{BFMS26} for $s\in[d-2,d)$ carries over with only minor modifications. In particular, the prefactor $\frac{1}{s}$ in the Riesz kernel $g_s$ must be retained, since $s$ may be negative when $d=2$.
 
Let $d\ge2$ and $s\in(d-3,d-2)$. We show that the measure $\mu_V$ defined in~\eqref{def of eq msr unconstraint} is a probability measure and satisfies the Euler--Lagrange variational conditions~\eqref{eq:EL condition}.  
Recall that the surface area of the unit sphere $\mathbb{S}^{d-1}\subset\mathbb{R}^d$ is given by 
\begin{equation}
\label{eq:surface area}
    |\mathbb{S}^{d-1}| = \frac{2\pi^{d/2}} {\Gamma(\frac{d}{2})}.
\end{equation}
Using polar coordinates and Euler's Beta integral, we readily see that $\mu_V$ is a probability measure:
\begin{align*}
    \int_{\R^{d}} \ud\mu_V(x) 
    &= \frac{\Gamma(\frac{s}{2}+2)}{\Gamma(2-\frac{d-s}{2})\Gamma(\frac{d}{2})} \int_0^1 2r^{d-1}(1-r^2)^{1-\frac{d-s}{2}} \ud r\\
    &= \frac{\Gamma(\frac{s}{2}+2)}{\Gamma(2-\frac{d-s}{2})\Gamma(\frac{d}{2})}\int_0^1 u^{\frac{d}{2}-1}(1-u)^{1-\frac{d-s}{2}} \ud u=1. 
\end{align*}

We note that it suffices to verify the Euler--Lagrange conditions for $\mathsf{C}=\mathsf{C}_{d,s}$, or equivalently, for $\mathsf{R}=1$, as the general case follows by scaling. 
To make the scaling argument precise, let us temporarily set $\mathsf{V}(x) =\mathsf{C}_{d,s}|x|^2$ as in \eqref{def of quad potential no constraint}, and suppose that the measure $\mu_{\mathsf{V}}$ given by~\eqref{def of eq msr unconstraint} satisfies the Euler--Lagrange conditions associated with $\mathsf{V}$. For an arbitrary $\mathsf{C}>0$, it follows from~\eqref{def of eq msr unconstraint} and~\eqref{def of radius R unconstraint} that
\[
    \frac{\ud\mu_V}{\ud x}(x) =\frac{1}{\mathsf{R}^d}\frac{\ud\mu_\mathsf{V}}{\ud x}(\mathsf{R}^{-1}x), \quad \mathsf{C}= \mathsf{R}^{-s-2} \mathsf{C}_{d,s}.
\]
Consequently, upon making the change of variables $x=\mathsf{R}x'$, we obtain 
\begin{align*}
    \frac{2}{s}\int_{\R^d} \frac{\ud \mu_V(x)}{|x-y|^s} +V(y) &= \frac{2}{s}\int_{\R^d} \frac{\ud\mu_\mathsf{V}(x')}{|\mathsf{R} x'-y|^s} + \mathsf{C} |y|^2 = \mathsf{R}^{-s}\bigg(\frac{2}{s}\int_{\R^d} \frac{\ud\mu_\mathsf{V}(x')}{|x'- \mathsf{R}^{-1}y|^s} + \mathsf{C}_{d,s}\Big|\frac{y}{\mathsf{R}}\Big|^2\bigg)
    \begin{cases}
        =c, &|y|\le \mathsf{R},\\
        \ge c, & |y|>\mathsf{R},
    \end{cases}
\end{align*}
for some constant $c$, where the last step follows from the Euler--Lagrange conditions for $\mu_{\mathsf{V}}$ with respect to $\mathsf{V}$. Hence, $\mu_V$ satisfies the Euler--Lagrange conditions for $V$ and is therefore the equilibrium measure associated with $V$.

Thus, it suffices to establish the Euler--Lagrange conditions in the normalised case $\mathsf{C}= \mathsf{C}_{d,s}$ and $\mathsf{R}=1$. 
As in~\cite{CSW22,CSW23,BFMS26}, the starting point is to reduce the $d$-dimensional integral to a one-dimensional one. For this, we apply the Funk--Hecke formula; see, e.g.,~\cite[Appendix~A.2]{CSW22}. This formula is valid for $s\neq 0$, which covers $s\in(d-3,d-2)$ when $d\ge 3$, while for $d=2$, the remaining case $s=0$ corresponds to the classical Coulomb case. Proceeding analogously to~\cite[Lemma~3.1]{BFMS26}, we obtain 
\begin{equation}
\label{eq:Riesz potential hypergeometric}
    \frac{2}{s}\int_{\R^d}\frac{\ud\mu_\mathsf{V}(x)}{|x-y|^s} =\frac{2}{s} \frac{\Gamma(\frac{s}{2}+2)}{\Gamma(2-\frac{d-s}{2})\Gamma(\frac{d}{2})}
    \int_0^1 {}_2F_1(\tfrac{s}{2},\tfrac{d-1}{2};d-1;\tfrac{4r|y|}{(|y|+r)^2})\frac{(1-r^2)^{1-\frac{d-s}{2}} }{(|y|+r)^s} (2r^{d-1})\ud r. 
\end{equation} 
Here, ${}_2F_1$ denotes the Gauss hypergeometric function \cite[Chapter 15]{NIST} defined by
\begin{equation}
\label{def of 2F1 ftn}
    {}_2F_1(a,b;c;z)
    :=
    \sum_{k=0}^{\infty}
    \frac{(a)_k(b)_k}{(c)_k}
    \frac{z^k}{k!},
    \qquad |z|<1,
\end{equation}
and by analytic continuation elsewhere, where $(a)_k:=\Gamma(a+k)/\Gamma(a)$ denotes the Pochhammer symbol.

In addition, it follows from~\cite{CSW22,BFMS26} or~\cite[Theorem~1.2]{GCO23} that, for $|y|<1$,
\begin{equation} \label{EL equality part unconstrained}
    \frac{2}{s}\int_{\R^d}\frac{\ud\mu_\mathsf{V}(x)}{|x-y|^s} =  \frac{\Gamma(\frac{s}{2}+2)\Gamma(\frac{d-s}{2})}{\frac{s}{2}\Gamma(\frac{d}{2})}
    {}_2F_1(\tfrac{s}{2}, -1; \tfrac{d}{2}; |y|^2)= \frac{\Gamma(\frac{s}{2}+2)\Gamma(\frac{d-s}{2})}{\frac{s}{2}\Gamma(\frac{d}{2})}-\mathsf{V}(y).
\end{equation}
Hence, the Euler--Lagrange equality holds throughout the support of the equilibrium measure.

To verify the Euler--Lagrange inequality, we make use of the quadratic
transformation~\cite[Eq.~(15.8.21)]{NIST} 
\begin{equation} \label{def of quad transform 2F1}
    {}_2F_1(\tfrac{s}{2},\tfrac{d-1}{2};d-1;\tfrac{4z}{(1+z)^2})= (1+z)^s {}_2F_1(\tfrac{s}{2}, 1-\tfrac{d-s}{2}; \tfrac{d}{2}; z^2), \qquad 0\le z<1.
\end{equation}  
For $|y|>1$, it follows from \eqref{def of quad transform 2F1} that
\begin{align}
\label{eq: Riesz potential |y|>1}
\begin{split}
    \frac{2}{s}\int_{\R^d}\frac{\ud\mu_\mathsf{V}(x)}{|x-y|^s} &=\frac{2}{s}  \frac{\Gamma(\frac{s}{2}+2)}{\Gamma(2-\frac{d-s}{2})\Gamma(\frac{d}{2})}|y|^{ -s }\int_0^1 {}_2F_1(\tfrac{s}{2}, 1-\tfrac{d-s}{2} ;\tfrac{d}{2}; \tfrac{r^2}{|y|^2}) (1-r^2)^{1-\frac{d-s}{2}} (2r^{d-1})\ud r\\
    &= \frac{2}{s} \frac{\Gamma(\frac{s}{2}+2)}{\Gamma(2-\frac{d-s}{2})\Gamma(\frac{d}{2})}|y|^{ -s}\int_0^1 {}_2F_1(\tfrac{s}{2}, 1-\tfrac{d-s}{2};\tfrac{d}{2}; \tfrac{u}{|y|^2}) (1-u)^{1-\frac{d-s}{2}} u^{\frac{d}{2}-1}\ud u\\
    &= \frac{2}{s}|y|^{ -s } {}_2F_1(\tfrac{s}{2},1-\tfrac{d-s}{2} ;\tfrac{s}{2}+2;|y|^{-2}),
\end{split}
\end{align}  
where we have used an integral representation of the hypergeometric function~\cite[Eq. (15.6.8)]{NIST} for the last identity. 

By setting $t=|y|^{-2}$ in~\eqref{eq: Riesz potential |y|>1},
the Euler--Lagrange inequality, in view of~\eqref{EL equality part unconstrained}, reduces to 
\begin{equation}
\label{eq:EL inequality reduced}
    \frac{2}{s}t^{ \frac{s}{2} } {}_2F_1(\tfrac{s}{2},1-\tfrac{d-s}{2} ;\tfrac{s}{2}+2;t) \ge \frac{\Gamma(\frac{s}{2}+2)\Gamma(\frac{d-s}{2})}{\frac{s}{2}\Gamma(\frac{d}{2})}
    {}_2F_1(\tfrac{s}{2}, -1; \tfrac{d}{2}; \tfrac{1}{t}), \qquad t \in (0, 1], 
\end{equation}
with equality if and only if $t=1$. Following the same steps as in~\cite[Lemma~3.5]{BFMS26}, the inequality~\eqref{eq:EL inequality reduced} is equivalent to 
\begin{equation}
    \frac{2}{s} \frac{\pi}{\sin(\tfrac{d-s}{2}\pi)}\frac{(1-t)^{1+\frac{d-s}{2}}}{\Gamma(\frac{s}{2})\Gamma(\tfrac{s-d}{2}+1)}
    t^{\frac{s}{2}}{}_2F_1(2, 1+\tfrac{d}{2}; 2+\tfrac{d-s}{2}; 1-t) >0,\qquad t\in (0,1).
\end{equation}
Since $s\in(d-3,d-2)$ and $d\ge2$, we have $s>-1$. By the reflection formula and the functional equation for the Gamma function, 
\begin{equation}
    \frac{\pi}{\sin(\frac{d-s}{2}\pi)}\frac{1}{\Gamma(\tfrac{s-d}{2}+1)}= \Gamma(\tfrac{d-s}{2}) >0, \qquad \frac{2}{s}\frac{1}{\Gamma(\frac{s}{2})}=\frac{1}{\Gamma(\tfrac{s}{2}+1)}>0.
\end{equation}
Moreover, since all the parameters  appearing in its defining series \eqref{def of 2F1 ftn} are positive,
\begin{equation}
    t^{\frac{s}{2}}{}_2F_1(2, 1+\tfrac{d}{2}; 2+\tfrac{d-s}{2}; 1-t) >0, \qquad t\in (0,1).
\end{equation}
It follows that the inequality \eqref{eq:EL inequality reduced} holds, which completes the proof. 
\end{proof}

The remainder of this section is devoted to the analysis of the auxiliary function $f:(0,\infty)\to\R$ defined by
\begin{equation}\label{def:f(x)}
    f(x) := x^{1-2\alpha}(1+x^2)^{-\frac{d-3}{2}} +\Big(\frac{s}{2}x +\frac{\alpha}{x}\Big) \mathrm{B}_{\frac{x^2}{1+x^2}}\Big(\frac{d-s-1}{2}, \frac{s}{2}\Big)-\mathsf{C}_{d-1, s}x, 
\end{equation}
where we recall that $\mathrm{B}_z$ denotes the incomplete beta function ~\eqref{def of incomplete beta}, while $\mathsf{C}_{d-1,s}$ and $\alpha$ are given by~\eqref{def of quad potential no constraint} and~\eqref{def of alpha}, respectively. 
Differentiating~\eqref{def:f(x)} and using~\eqref{def of incomplete beta},
we readily obtain, for $x>0$, 
\begin{align}
\label{eq:f'(x)}
\begin{split}
    f'(x) 
    &= x^{-2\alpha}(1+x^2)^{-\frac{d-3}{2}} + \Big(\frac{s}{2}-\frac{\alpha}{x^2}\Big)\mathrm{B}_{\frac{x^2}{1+x^2}}\Big(\frac{d-s-1}{2},\frac{s}{2}\Big)-\mathsf{C}_{d-1,s}.
\end{split}
\end{align}
Thus, $\mathsf{x}$ solves~\eqref{eq:critical point incomplete beta} if and only if $f'(\mathsf{x})=0$. 
We next establish the relevant properties of $f$ and, in particular, prove the existence and uniqueness
of the solution to~\eqref{eq:critical point incomplete beta} for
$s\in(d-3,d-2)$ by showing that $f$ has a unique global maximum on
$[0,\infty)$.

\begin{lem} \label{lem: property f(x)}
Assume that $s\in(d-3,d-1)$. The function $f$ defined in~\eqref{def:f(x)} satisfies the following properties: 
\begin{enumerate}[label=\textup{(\roman*)}]
    \item \label{item: limits f(x)}
  As $x\to\infty$, we have $f(x)\to-\infty$, whereas 
    \begin{equation}
         \lim_{x\to 0^+}f(x)=  
        \begin{cases}
            0, &\textup{if } s \in (d-3,d-2),
            \smallskip 
            \\
            2, &\textup{if } s= d-2,
            \smallskip 
            \\
            +\infty, &\textup{if } s\in (d-2,d-1).
        \end{cases}
    \end{equation}
    Consequently, $f$ attains a global maximum on $[0,\infty)$ for $s\in(d-3,d-2]$, whereas no finite global maximum exists for $s\in(d-2,d-1)$. 
    \smallskip 
    \item \label{item: maximum existence}
    In the strongly long-ranged regime $s\in(d-3,d-2]$, the function $f$ attains its unique global maximum at $x=\mathsf{x}$, where $\mathsf{x}$ is given by~\eqref{eq:critical point incomplete beta} for $s\in (d-3,d-2)$ and $\mathsf{x}=0$ for $s=d-2$.
    Moreover, the maximum value is given by
    \begin{equation}
    \label{eq:f(x) maximum characterisation}
        \max_{x\in [0,\infty)} f(x) =f(\mathsf{x})= \frac{2\alpha}{\mathsf{x}} \mathrm{B}_{\frac{\mathsf{x}^2}{1+\mathsf{x}^2}}\Big(\frac{d-s-1}{2},\frac{s}{2}\Big),
    \end{equation}
   where, in the case $\mathsf{x}=0$, the right-hand side is understood in
the limiting sense as $\mathsf{x}\to0^+$. 
\end{enumerate}
\end{lem}

By~\eqref{eq:f(x) maximum characterisation}, the critical value
$a_{\rm cri}(d,s)$ in~\eqref{eq:a critical incomplete beta} admits the
variational representation
\begin{equation} \label{a critical value in terms of max f}
    a_{\rm cri}(d,s) = \frac{ \frac{s}{2}+1 }{2\alpha R^{s+1}} \frac{1}{\mathsf{C}_{d,s} } f(\mathsf{x})=   \frac{ \frac{s}{2}+1 }{2\alpha R^{s+1}} \frac{1}{\mathsf{C}_{d,s} } \max_{x\in [0,\infty)}f(x),
\end{equation}
where $\mathsf{C}_{d,s}$ is given by \eqref{def of quad potential no constraint}. 

We now prove Lemma~\ref{lem: property f(x)}.

\begin{proof}[Proof of Lemma~\ref{lem: property f(x)}]
We first show \textrm{(i)}. Note that as $x\to\infty$, we have 
\begin{equation}\label{eq:Beta completion}
    \lim_{x\to \infty}\Big(\frac{s}{2}\mathrm{B}_{\frac{x^2}{1+x^2}}\Big(\frac{d-s-1}{2},\frac{s}{2}\Big)+x^{-s}\Big) 
    = \frac{\Gamma(\frac{d-s-1}{2})\Gamma(\frac{s}{2}+1)} {\Gamma(\frac{d-1}{2})}
    = \mathsf{C}_{d-1,s}\frac{d-1}{s+2},
\end{equation}
It therefore follows from~\eqref{def:f(x)} that, for all $s\in(d-3,d-1)$,
\begin{equation}
    \lim_{x\to\infty} \frac{f(x)}{x} =\lim_{x\to \infty}
    \Big(\frac{s}{2}\mathrm{B}_{\frac{x^2}{1+x^2}}\Big(\frac{d-s-1}{2},\frac{s}{2}\Big)+x^{-s}\Big) - \mathsf{C}_{d-1,s} = \mathsf{C}_{d-1,s} \frac{(d-3)-s}{s+2} <0.
\end{equation}
Consequently, $f(x)\to-\infty$ as $x\to\infty$. 

For the limit as $x\to0^+$, we first observe that 
\begin{equation}\label{eq:Beta to zero}
    \mathrm{B}_{\frac{x^2}{1+x^2}}\Big(\frac{d-s-1}{2},\frac{s}{2}\Big) =\int_0^{\frac{x^2}{1+x^2}} u^{-\alpha} (1-u)^{\frac{s}{2}-1}\ud u = \frac{x^{2-2\alpha}}{1-\alpha}\big(1+ O(x^2)\big), \quad x\to 0^+.
\end{equation}
For $s\in(d-3,d-2)$, we have $\alpha\in(0,\frac12)$, and hence \eqref{eq:Beta to zero} together with~\eqref{def:f(x)} readily gives $\lim_{x\to 0^+}f(x)=0$. On the other hand, for $s\in[d-2,d-1)$, we have
\begin{align*}
    \lim_{x\to0^+}\frac{f(x)}{x^{1-2\alpha}}
    &=
    \lim_{x\to0^+}
    \bigg[
        (1+x^2)^{-\frac{d-3}{2}}
        +
        \Big(\frac{s}{2}x+\frac{\alpha}{x}\Big)
        \frac{1}{x^{1-2\alpha}}
        \mathrm{B}_{\frac{x^2}{1+x^2}}
        \Big(
            \frac{d-s-1}{2},\frac{s}{2}
        \Big)
        -
        \mathsf{C}_{d-1,s}x^{2\alpha}
    \bigg]
    =
    \frac{1}{1-\alpha}.
\end{align*} 
Since $\alpha=\frac{1}{2}$ when $s=d-2$, while $\alpha\in(\frac{1}{2},1)$ when $s\in(d-2,d-1)$, the desired limits follow, completing the proof of~\ref{item: limits f(x)}.  

\medskip
Next, we prove~\textrm{(ii)}. We first note the useful representation 
\begin{equation}\label{eq:incomplete Beta integral}
    \mathrm{B}_{\frac{x^2}{1+x^2}}\Big(\frac{d-s-1}{2},\frac{s}{2}\Big) = \int_0^{x^2} v^{-\alpha}(1+v)^{-\frac{d-1}{2}}\ud v= x^{2-2\alpha}\int_0^1 u^{-\alpha}(1+x^2u)^{-\frac{d-1}{2}}\ud u.
\end{equation}
Using~\eqref{eq:incomplete Beta integral}, for $x>0$, we have 
\begin{align}\label{eq:f''(x)}
\begin{split}
    f''(x) 
    &= -2\alpha x^{-1-2\alpha}(1+x^2)^{-\frac{d-1}{2}}\bigg(2- \int_0^1  u^{-\alpha}\Big(\frac{1+x^2}{1+x^2u}\Big)^{\frac{d-1}{2}}\ud u\bigg).
\end{split}
\end{align}
Notice that the function
\[
h(x):= \int_0^1  u^{-\alpha}\Big(\frac{1+x^2}{1+x^2u}\Big)^{\frac{d-1}{2}}\ud u, \quad x\ge 0,
\]
is an increasing function in $x$. In particular, we have
\begin{equation}
    h(0) = \frac{1}{1-\alpha}\le 2,
\end{equation}
where the equality holds only if $s=d-2$. 

We first prove~\ref{item: maximum existence} in the case $s=d-2$. Since
$h(x)\ge h(0)=2$, it follows from~\eqref{eq:f''(x)} that 
\[
    f''(x) = -2\alpha x^{-1-2\alpha}(1+x^2)^{-\frac{d-1}{2}} (2-h(x) )\ge  0, \qquad x>0. 
\]
Thus, $f$ is convex on $(0,\infty)$, and hence $f'$ is increasing. On the
other hand, by~\eqref{eq:f'(x)} and~\eqref{eq:Beta completion}, 
\begin{align*}
    \lim_{x\to\infty} f'(x) &= \lim_{x\to\infty}\frac{s}{2}\mathrm{B}_{\frac{x^2}{1+x^2}}\Big(\frac{1}{2},\frac{s}{2}\Big) -\mathsf{C}_{d-1,d-2}=-\frac{1}{d}\mathsf{C}_{d-1,d-2}<0.
\end{align*}
Since $f'$ is increasing, we conclude that $f'(x)<0$ for all $x>0$.
Therefore, $f$ is strictly decreasing on $(0,\infty)$ and attains its
unique global maximum at $x=\mathsf{x}=0$. Using~\eqref{eq:Beta to zero}, we also obtain 
\[
    \lim_{x\to0^+}\frac{1}{x}\mathrm{B}_{\frac{x^2}{1+x^2}}\Big(\frac{1}{2},\frac{s}{2}\Big) = 2 = f(0),
\]
which shows (ii) for $s=d-2$.

We now turn to the case $s\in(d-3,d-2)$. By~\eqref{eq:f''(x)},
the inequality $h(0)<2$ implies that $f''(x)<0$ for all sufficiently
small $x>0$. Moreover, since $h$ is increasing, $f''$ has at most one
zero on $(0,\infty)$.

We claim that $f$ has a unique critical point on $(0,\infty)$. In view
of~\ref{item: limits f(x)} and the observation following~\eqref{eq:f'(x)},
this claim implies that~\eqref{eq:critical point incomplete beta} admits
a unique solution $\mathsf{x}\in(0,\infty)$ and that $f$ attains its
unique global maximum at $x=\mathsf{x}$.

We argue by contradiction. Since $f''$ has at most one zero on
$(0,\infty)$, the function $f$ has at most two critical points there.
Suppose that $f$ has exactly two critical points, say
$\mathsf{x}_1$ and $\mathsf{x}_2$, with $0<\mathsf{x}_1<\mathsf{x}_2$. Then it follows that $f''$ must vanish at
some point in $(\mathsf{x}_1,\mathsf{x}_2)$. Since $f''(x)<0$ for all
sufficiently small $x>0$ and $f''$ has at most one zero, it follows that
\[
    f''(\mathsf{x}_1)<0, \qquad f''(\mathsf{x}_2)>0.
\]
On the other hand, since $f''(x)>0$ for $x>\mathsf{x}_2$, the function
$f'$ is increasing on $(\mathsf{x}_2,\infty)$. Hence, by
\eqref{eq:f'(x)} and~\eqref{eq:Beta completion}, 
\[
    0=f'(\mathsf{x}_2)\le \lim_{x\to\infty}f'(x) = \lim_{x\to\infty}\frac{s}{2}\mathrm{B}_{\frac{x^2}{1+x^2}}\Big(\frac{d-s-1}{2},\frac{s}{2}\Big)-\mathsf{C}_{d-1,s} = \mathsf{C}_{d-1,s}\frac{(d-3)-s}{s+2}<0,
\]
which is a contradiction. Therefore, $f$ has a unique critical point
on $(0,\infty)$.  

It remains to verify~\eqref{eq:f(x) maximum characterisation} for
$s\in(d-3,d-2)$. From~\eqref{def:f(x)} and~\eqref{eq:f'(x)}, we readily obtain 
\begin{equation}
    f(x)-xf'(x)= \frac{2\alpha}{x}\mathrm{B}_{\frac{x^2}{1+x^2}}\Big(\frac{d-s-1}{2},\frac{s}{2}\Big).
\end{equation}
Evaluating this identity at the unique critical point $x=\mathsf{x}$,
for which $f'(\mathsf{x})=0$, gives the desired result \eqref{eq:f(x) maximum characterisation}. 
\end{proof}

\begin{rem}[Critical value for special cases] \label{Rem_critical special}
\label{rem:compute critical values}
We record some explicit formulas for $a_{\rm cri}(d,s)$ at special
parameter values, which follow readily from~\eqref{a critical value in terms of max f}.
For $s=d-2$, we have $\mathsf{x}=0$ and $f(0)=2$, and hence 
\[
    a_{\rm cri}(d,d-2)= \frac{1}{R^{d-1}} \frac{\Gamma(\frac{d}{2}+1)}{\Gamma(\frac{d}{2})} f(0) = d\Big(\frac{\Gamma(\frac{d+1}{2})}{\sqrt{\pi}\Gamma(\frac{d}{2}+1)}\Big)^{\frac{d-1}{d}},
\]
which recovers~\eqref{def of critical value Coulomb}. 
For $d=1$, by~\eqref{eq:incomplete Beta integral} and \eqref{def of alpha}, we have 
\[
    \mathrm{B}_{\frac{x^2}{1+x^2}}\Big(-\frac{s}{2},\frac{s}{2}\Big)= x^{2-2\alpha}\int_0^1 u^{-\alpha}\ud u =\frac{x^{2-2\alpha}}{1-\alpha} = -\frac{2}{s}x^{2-2\alpha}.
\]
Consequently,~\eqref{eq:critical point incomplete beta} reduces to 
\[
    x^{-2\alpha}(1+x^2) -x^{2-2\alpha}+\frac{2\alpha}{s}x^{-2\alpha} = \frac{2(s+1)}{s}x^{-(s+2)}= \mathsf{C}_{0,s}.
\]
Recalling that $s\in(-2,-1)$ when $d=1$, we obtain 
\[
    \mathsf{x}= \Big(\frac{s}{2(s+1)}\mathsf{C}_{0,s}\Big)^{-\frac{1}{s+2}}, \qquad \frac{1}{\mathsf{x}}\mathrm{B}_{\frac{x^2}{1+x^2}}\Big(-\frac{s}{2},\frac{s}{2}\Big)=\frac{2}{|s|}\mathsf{x}^{-s-1} = \frac{2}{|s|}\Big(\frac{s}{2(s+1)}\mathsf{C}_{0,s}\Big)^{\frac{s+1}{s+2}}.
\]
Therefore, it follows that 
\begin{align*}
    a_{\rm cri}(1,s) &=\Big(\frac{\mathsf{C}_{1,s}}{\mathsf{C}_{0,s}}\Big)^{\frac{s+1}{s+2}}
    \frac{\frac{s}{2}+1}{\mathsf{C}_{1,s}}   \frac{2}{|s|}\Big(\frac{s}{2(s+1)}\mathsf{C}_{0,s}\Big)^{\frac{s+1}{s+2}}
    = \frac{s+2}{|s|}\mathsf{C}_{1,s}^{-1} \Big(\frac{s}{2(s+1)}\mathsf{C}_{1,s}\Big)^{\frac{s+1}{s+2}}\\
    &=\frac{s+2}{2|s+1|}\Big(\frac{s}{2(s+1)}\mathsf{C}_{1,s}\Big)^{-\frac{1}{s+2}}= \frac{s+2}{2|s+1|}\Big(\frac{2(s+1)}{s}\frac{\Gamma(\frac{3}{2})}{\Gamma(\frac{1-s}{2})\Gamma(\frac{s}{2}+2)}\Big)^{\frac{1}{s+2}},
\end{align*}
which recovers~\eqref{def of critical value d=1}.
\end{rem}

We conclude this section by deriving an alternative representation of
$f$~\eqref{def:f(x)} in terms of the Gauss hypergeometric function
${}_2F_1$. This representation will play a crucial role in verifying
the Euler--Lagrange variational conditions in
Section~\ref{sec: proof of main thm}. 

\begin{lem}\label{lem:f(x) hypergeometric}
For $s\in(d-3,d-1)$, the function $f$ in \eqref{def:f(x)} admits the representation
\begin{equation}
\label{eq:f(x) hypergeometric}
    f(x)= \frac{x^{1-2\alpha}}{1-\alpha} {}_2F_1 (\tfrac{d-1}{2},-\alpha; 2-\alpha ;-x^2) - \mathsf{C}_{d-1,s}x. 
\end{equation}
\end{lem}

\begin{proof}
We begin with the representation~\eqref{eq:incomplete Beta integral},
which, by the integral representation of the Gauss hypergeometric
function (see, e.g., \cite[Eq.~(15.6.1)]{NIST})
\begin{equation}
\label{def of 2F1 Euler integral}
    \int_0^1
    u^{b-1}(1-u)^{c-b-1}(1-zu)^{-a}\,\ud u
    =
    \frac{\Gamma(b)\Gamma(c-b)}{\Gamma(c)}
    {}_2F_1(a,b;c;z),
    \qquad
    (\re c>\re b>0)
\end{equation}
gives
\begin{equation}
    \mathrm{B}_{\frac{x^2}{1+x^2}}\Big(\frac{d-s-1}{2},\frac{s}{2}\Big) = x^{2-2\alpha}\int_0^1 u^{-\alpha}(1+x^2u)^{-\frac{d-1}{2}} \ud u = \frac{x^{2-2\alpha}}{1-\alpha}{}_2F_1(\tfrac{d-1}{2},1-\alpha; 2-\alpha ;-x^2). 
\end{equation}
Next, by the generalised binomial theorem and \cite[Eq.~(15.4.6)]{NIST}, we have 
\begin{equation}
    (1+x^2)^{-\frac{d-3}{2}} = (1+x^2)\cdot (1+x^2)^{-\frac{d-1}{2}}= (1+x^2){}_2F_1(\tfrac{d-1}{2}, 2-\alpha; 2-\alpha; -x^2).
\end{equation}
Finally, Gauss's contiguous relation~\cite[Eq.~(15.5.11)]{NIST} yields
\[
    (1-\alpha)(1+x^2){}_2F_1(\tfrac{d-1}{2},2-\alpha;2-\alpha;-x^2) + \Big(\frac{s}{2}x^2+\alpha\Big) {}_2F_1(\tfrac{d-1}{2},1-\alpha;2-\alpha;-x^2)
    =  {}_2F_1(\tfrac{d-1}{2}, -\alpha;2-\alpha ;-x^2). 
\]
Substituting these identities into~\eqref{def:f(x)} gives the desired
representation~\eqref{eq:f(x) hypergeometric}. 
\end{proof}

\section{Proof of Theorem \texorpdfstring{\ref{thm: fully confined}}{2.2}}
\label{sec: proof of main thm}

This section is devoted to the proof of Theorem~\ref{thm: fully confined}.
It suffices to consider the case of nonzero $s\in(d-3,d-1)$. Indeed,
as mentioned above, half-space constrained equilibrium measures in the
Coulomb case have already been characterised in~\cite{ASZ14,FIT26}.
On the other hand, for $s\in[d-1,d)$, the Riesz interaction becomes
hypersingular when restricted to the hyperplane $\{x_d=a\}$, so that
any probability measure supported on this hyperplane has infinite
Riesz energy; see~\cite{ADKKMMS19,HLSS18}. Consequently, the
Euler--Lagrange equality  \eqref{eq:EL condition} cannot hold with a finite Robin constant, and hence a fully confined phase is impossible in this regime.

We begin by observing that the half-space constrained equilibrium measure
$\widehat{\mu}_a$ is fully confined to the hyperplane $\{x_d=a\}$ if and
only if a measure of the form
\begin{equation} \label{eq: mu-hat-a confined ansatz}
    \ud \widehat\mu_a(x) = \begin{cases}
        \ud \nu(\hat x) \otimes \delta_a(x_d), & \text{if } d \geq 2,
        \smallskip 
        \\
        \delta_a(x), & \text{if } d =1
    \end{cases}
\end{equation}
satisfies the Euler--Lagrange conditions~\eqref{eq:EL condition} for some
probability measure $\nu$ on $\R^{d-1}$. 
For $d\ge2$, the Euler--Lagrange equality restricted to the hyperplane
$\{x_d=a\}$ uniquely determines $\nu$ as 
\begin{equation} \label{eq: ansatz measure nu}
    \ud \nu(\hat x) = \frac{\Gamma(\frac{s}{2}+2)}{\Gamma(\alpha+1)\Gamma(\frac{d+1}{2})}\Big(1-\frac{|\hat x|^2}{R^2}\Big)^{\alpha} \cdot \mathbbm{1}_{|\hat x|\le R} \frac{\Gamma(\frac{d+1}{2})}{(\sqrt{\pi}R)^{d-1}}\ud \hat x,
\end{equation}
where $\alpha$ and $R$ are given by~\eqref{def of alpha} and
\eqref{def of R}, respectively. Indeed, $\nu$ is precisely the
equilibrium measure associated with the $(d-1)$-dimensional external
potential 
\begin{equation}
    W(\hat{x}) := V(\hat x,a)-V(0,a) = V(\hat x,0)= \mathsf{C}_{d,s} \, |\hat x|^2.
\end{equation}
Consequently, the Euler--Lagrange equality for~\eqref{eq: mu-hat-a confined ansatz} on the hyperplane $\{x_d=a\}$ follows directly from the characterisation of the unconstrained equilibrium measure (Proposition~\ref{Prop_eq msr unconstrained}).

The main challenge in determining whether the half-space constrained
equilibrium measure is fully confined lies in verifying the
Euler--Lagrange inequality~\eqref{eq:EL condition} in the interior
$\{x_d>a\}$. To this end, for $a\in\R$ and $t\ge0$, we define 
\begin{align} \label{eq: def F(x,t)}
    F(\hat{x},t) := \begin{dcases}
        2\int_{\R^{d-1}} g_s((\hat{x}-y, t)) \ud\nu(y)+V(\hat{x},a+t), & \text{if } d\geq 2,
        \smallskip 
        \\
        \frac{2}{s} t^{-s} + V(a+t) &\text{if } d = 1.
    \end{dcases}
\end{align}
The Euler--Lagrange conditions~\eqref{eq:EL condition} for the candidate
fully confined measure are equivalent to 
\begin{equation} \label{eq: EL via F(x,t)}
    F(\hat x,t) 
    \begin{cases}
        = F(0,0), &\textup{if } |\hat x|\le R, \;t=0,
        \smallskip 
        \\
        \ge F(0,0), &\text{otherwise},
    \end{cases}
\end{equation}
where we have used the fact that $(0,a)$ belongs to the support of the
candidate measure. By construction, the equality in~\eqref{eq: EL via F(x,t)}
on the hyperplane $t=0$ has already been established. It therefore remains
to verify the inequality for $t>0$. We begin by deriving an explicit
formula for $F(0,t)$ as an analogue of \cite[Lemma 4.1]{BFMS26}.

\begin{lem}\label{lem:F(0,t)}
For nonzero $s\in(d-3,d-1)$ and $t\ge0$, we have
\begin{align}\label{eq:F(0,t)}
\begin{split}
    F(0,t) &=\frac{(\frac{s}{2}+1)}{\alpha(\alpha-1)} \frac{t^{d-s-1}}{R^{d-1}} {}_2F_1 (\tfrac{d-1}{2},-\alpha; 2-\alpha ;-t^2/R^2) + \mathsf{C}_{d,s}\Big(a^2 + \frac{d-1}{s}R^2  + 2at + \frac{s+2}{2\alpha}t^2\Big),
\end{split}
\end{align}
where $\mathsf{C}_{d,s}$ is given by \eqref{def of quad potential no constraint}. 
In particular, we have
\begin{equation}\label{eq:F(0,0)}
    F(0,0)=\mathsf{C}_{d,s}  \Big(a^2+\frac{d-1}{s}R^2\Big). 
\end{equation}
\end{lem}

\begin{proof}
The assertion is immediate for $d=1$, since
${}_2F_1(0,b;c;z)=1$. We therefore assume $d\ge2$ in what follows.

It follows from~\eqref{def of interaction g} and~\eqref{eq: def F(x,t)}
that 
\begin{equation}
    F(0,t) = \frac{2}{s} \int_{\R^{d-1}} \frac{\ud \nu(y)}{(|y|^2 + t^2)^{s/2}} + V_a(0,a+t).
\end{equation}
Using~\eqref{eq: ansatz measure nu} and polar coordinates, we obtain
\begin{align*}
    \int_{\R^{d-1}}\frac{\ud\nu(y)}{(|y|^2+t^2)^{s/2}} &= \frac{\Gamma(\frac{s}{2}+2)}{\Gamma(\alpha+1)\Gamma(\frac{d+1}{2})}\int_{|y|\le 1} \frac{(1-|y|^2)^\alpha}{(R^2|y|^2+t^2)^{s/2}} \frac{\Gamma(\frac{d+1}{2})}{\pi^{\frac{d-1}{2}}}\ud y\\
    &= \frac{\Gamma(\frac{s}{2}+2)(d-1)}{\Gamma(\alpha+1)\Gamma(\frac{d+1}{2})}\int_0^1 \frac{(1-r^2)^\alpha}{(R^2r^2+t^2)^{s/2}}r^{d-2}\ud r\\
    &= \frac{\Gamma(\frac{s}{2}+2)(d-1)}{\Gamma(\alpha+1)\Gamma(\frac{d+1}{2})}\frac{1}{t^s}\int_0^1 \Big(1+\frac{R^2r^2}{t^2}\Big)^{-s/2} (1-r^2)^\alpha r^{d-2}\ud r.  
\end{align*}
Therefore, by the Euler integral representation \eqref{def of 2F1 Euler integral} of the Gauss hypergeometric
function, we obtain 
\begin{equation}\label{eq:F(0,t) inter}
    F(0,t) =\frac{2}{st^s}{}_2F_1(\tfrac{d-1}{2},\tfrac{s}{2}; \tfrac{s}{2}+2 ;- R^2/t^2) + \mathsf{C}_{d,s} (a+t)^2 .
\end{equation}

We further recall the linear transformation formula for the Gauss
hypergeometric function (see \cite[Eq.~(15.8.2)]{NIST})
\begin{align}
\begin{split}
    {}_2F_1(a,b;c;-z) &= \frac{\Gamma(b-a)\Gamma(c)z^{-a}}{\Gamma(b)\Gamma(c-a)} {}_2F_1(a, a-c+1;a-b+1;-1/z)\\
    &\quad + \frac{\Gamma(a-b)\Gamma(c)z^{-b}}{\Gamma(a)\Gamma(c-b)} {}_2F_1(b,b-c+1; b-a+1; -1/z),
\end{split}
\end{align}
valid for $z>0$ and $b-a\notin\mathbb{Z}$. Substituting $a=\frac{d-1}{2}$, $b=\frac{s}{2}$, and $c=\frac{s}{2}+2$, we obtain
\begin{align*}
    {}_2F_1(\tfrac{d-1}{2},\tfrac{s}{2};\tfrac{s}{2}+2; -R^2/t^2)
    &= \frac{\frac{s}{2}(\frac{s}{2}+1)}{\alpha(\alpha-1)}\frac{t^{d-1}}{R^{d-1}} {}_2F_1(\tfrac{d-1}{2}, -\alpha; 2-\alpha; -t^2/R^2)\\
    &\quad +\frac{\Gamma(1-\alpha)\Gamma(\frac{s}{2}+2)}{\Gamma(\frac{d-1}{2})}\frac{t^s}{R^s}\Big(1+\frac{s}{2\alpha}\frac{t^2}{R^2}\Big),
\end{align*}
where in the last term we have used the terminating case of the Gauss hypergeometric function~\cite[Eq.~(15.2.4)]{NIST}. Substituting this into~\eqref{eq:F(0,t) inter} yields the desired formula \eqref{eq:F(0,t)}. 
\end{proof}

We are now ready to prove Theorem~\ref{thm: fully confined} in the
weakly long-ranged regime.

\begin{proof}[Proof of Theorem~\ref{thm: fully confined}~\ref{thm: fully confined weak}]
 It suffices to consider $s\in(d-2,d-1)$. Indeed, as noted above, for
$s\in[d-1,d)$ the Riesz interaction becomes hypersingular when restricted
to the hyperplane $\{x_d=a\}$, so that a fully confined equilibrium
measure cannot occur.

By Lemma~\ref{lem:F(0,t)}, as $t\downarrow 0$, we have 
    \begin{equation}
        F(0,t)-F(0,0) = \frac{s+2}{2\alpha (\alpha-1)}\frac{t^{d-s-1}}{R^{d-1}} + O(t). 
    \end{equation} 
Recall from~\eqref{def of alpha} that $\alpha\in(1/2,1)$. Since the coefficient of the leading-order term is negative, for all sufficiently small $t>0$, we have $F(0,t)<F(0,0)$. Thus, the Euler--Lagrange inequality fails for the candidate fully confined measure $\nu\otimes\delta_a$.
Consequently, the equilibrium measure $\widehat{\mu}_a$ can never be
fully confined to the hyperplane $\{x_d=a\}$.
\end{proof}

We now turn to the strongly long-ranged regime $s\in(d-3,d-2]$.
In this case, we need to verify the Euler--Lagrange inequality~\eqref{eq: EL via F(x,t)} for all $\hat{x}\in\R^{d-1}$ and $t\ge0$, which would be cumbersome to establish directly. We circumvent this difficulty by showing that it suffices to verify the inequality along $\hat{x}=0$. For this purpose, we use the following lemma, which is a direct consequence of~\cite[Lemma~2.2 (i)]{Li83}. 

\begin{lem}\label{lem:convolution}
Let $d\ge1$, and let $\phi,\psi:\R^d\to[0,\infty]$ be measurable,
radially symmetric, and radially non-increasing. Suppose that the
convolution
\begin{equation}
  (\phi*\psi)(x)
    :=
    \int_{\R^d}\phi(x-y)\psi(y)\,\ud y
\end{equation} 
is well defined and finite for every $x\in\R^d$. Then $\phi*\psi$ is
also radially symmetric and radially non-increasing.
\end{lem}

\begin{proof}[Proof of Theorem~\ref{thm: fully confined}~\ref{thm: fully confined strong}] 
We first examine the Euler--Lagrange inequality~\eqref{eq: EL via F(x,t)}
along $\hat{x}=0$. By Lemmas~\ref{lem:f(x) hypergeometric}
and~\ref{lem:F(0,t)}, we have  
    \begin{equation}
    \begin{split}
        F(0,t)-F(0,0) &=\frac{\frac{s}{2}+1}{\alpha(\alpha-1)} \frac{t^{d-s-1}}{R^{d-1}} {}_2F_1 (\tfrac{d-1}{2},-\alpha; 2-\alpha ;-t^2/R^2)
        + \mathsf{C}_{d,s} \Big(2at + \frac{s+2}{2\alpha}t^2\Big)
        \\
        &=  \Big(\frac{s}{2}+1 \Big) \, t \, \Big( \mathsf{C}_{d,s} \frac{4a}{s+2} - \frac{1}{\alpha R^{s+1}}f(t/R) \Big). 
    \end{split}
    \end{equation} 
    Here, we recall that $\alpha$ is given by \eqref{def of alpha}.  
    Since $t/R$ ranges over $[0,\infty)$ as $t$ ranges over
$[0,\infty)$, Lemma~\ref{lem: property f(x)} implies that 
\begin{equation} \label{F(0,t) vs F(0,0)}
    F(0,t)\ge F(0,0)
    \qquad\text{for all }t\ge0
\end{equation} 
if and only if
\[
  \mathsf{C}_{d,s} \frac{4a}{s+2} 
    \ge
    \frac{1}{\alpha R^{s+1}}
    \max_{x\in[0,\infty)}f(x).
\]
By~\eqref{a critical value in terms of max f}, this is equivalent to $   a\ge a_{\rm cri}(d,s).$  
Therefore, the Euler--Lagrange inequality along $\hat{x}=0$ holds for all $t\ge0$ if and only if $a\ge a_{\rm cri}(d,s)$. For $d=1$, this already completes the proof of the theorem.

  It remains to prove the assertion for $d\ge2$. For this, all we need to show is 
  \begin{equation} \label{F(x,t) vs F(0,0)}
      F(\hat{x},t)\ge F(0,0)
    \qquad\text{for all }\hat{x}\in\R^{d-1},\quad t>0.
  \end{equation} 
Once this claim is established, the result for $d\ge2$ follows
immediately. 
To prove the claim, observe that, for $t>0$,  
    \begin{align}
    \frac{\partial}{\partial t}
    \bigg[
        \frac{2}{s}
        \int_{\R^{d-1}}
        \frac{\ud\nu(y)}
        {(|\hat{x}-y|^2+t^2)^{s/2}}
    \bigg]
    &=
    -\int_{\R^{d-1}}
    \frac{2t}
    {(|\hat{x}-y|^2+t^2)^{\frac{s}{2}+1}}
    \,\ud\nu(y)
    =
    -(\phi*\psi_t)(\hat{x}),
\end{align}
    where
    \begin{equation}
        \phi(y) = \frac{\ud\nu}{\ud y}(y), \qquad \psi_t(y) = \frac{2t}{(|y|^2+t^2)^{\frac{s}{2}+1}}. 
    \end{equation}
   Both $\phi$ and $\psi_t$ are nonnegative, radially symmetric, and
radially non-increasing. Hence, by Lemma~\ref{lem:convolution},
$\phi*\psi_t$ is radially non-increasing, and therefore
\begin{equation}
    (\phi*\psi_t)(\hat{x})
    \le
    (\phi*\psi_t)(0).
\end{equation}
It follows that 
\begin{align*}
    F(\hat{x},t)-F(\hat{x},0)
    &= \int_0^t \frac{\partial}{\partial t'}\bigg(\frac{2}{s}\int_{\R^{d-1}}\frac{\ud\nu(y)}{(|\hat{x}-y|^2+t'^2)^{s/2}}\bigg) \ud t' +V(\hat{x},a+t)-V(\hat{x},a)
    \\
    &=
    -\int_0^t
    (\phi*\psi_{t'})(\hat{x})\,\ud t'
    +V(\hat{x},a+t)-V(\hat{x},a)
    \\
    &\ge
    -\int_0^t
    (\phi*\psi_{t'})(0)\,\ud t'
    +V(0,a+t)-V(0,a)
    =
    F(0,t)-F(0,0)
    \ge0,
\end{align*} 
  where we have used \eqref{F(0,t) vs F(0,0)} and  
  $$
  V(\hat{x},a+t)-V(\hat{x},a)
    =
    V(0,a+t)-V(0,a).
  $$   
Consequently,
\[
    F(\hat{x},t)\ge F(\hat{x},0)\ge F(0,0),
\]
where the second inequality follows from the Euler--Lagrange inequality
on the hyperplane $\{x_d=a\}$. This proves~\eqref{F(x,t) vs F(0,0)}
and completes the proof. 
\end{proof}



\end{document}